\documentclass[11pt]{article}
\usepackage{amsmath,amssymb,amsthm,color,graphicx,epstopdf,tikz}
\usepackage[shortlabels]{enumitem}

\numberwithin{equation}{section}

\newcommand{\be}{\begin{equation}\begin{gathered}} \newcommand{\ee}{\end{gathered}\end{equation}}
\newcommand{\la}{\langle} \newcommand{\ra}{\rangle}
\newcommand{\al}{\alpha}
\newcommand{\bt}{{\boldsymbol{\theta}}}
\newcommand{\bs}{\boldsymbol}
\newcommand{\hwk}{\hat W^{(k)}}
\newcommand{\cN}{{\mathcal{N}}}
\newcommand{\fs}{{\mathfrak{S}}}

 \newtheorem {hypo}{\bf\hspace{-\parindent}Hypothesis}[section]

 \newtheorem {ass}[hypo]{Assumption}
 \newtheorem {prop}[hypo]{Proposition}
 \newtheorem {lem}[hypo]{Lemma}
 \newtheorem {theo}[hypo]{Theorem}
 \newtheorem {defn}[hypo]{Definition}

 \theoremstyle{remark}

  \newcommand\lb{\left(}
  \newcommand\rb{\right)} 
  
   \newcommand\Cb{\mathbb{C}}

 \newcommand{\address}[1] {%
 \begingroup\begin{center} #1 \end{center}\endgroup}
 
 \renewcommand{\title}[1] {%
  \begingroup\begin{center}\vspace{0.0cm}\LARGE
  \addtolength{\baselineskip}{1mm} #1 \end{center}\endgroup}
 
  \renewcommand{\author}[1] {%
  \begingroup\begin{center}\vspace{0.2cm} \large #1 \vspace{0.2cm}
  \end{center}\endgroup}

\begin{document}

\title{Higher rank isomonodromic deformations and $W$-algebras}
\date{}
\author{ P.~Gavrylenko$^{a,e,f}$\footnote{pasha145@gmail.com}, N.~Iorgov$^{a,c,d}$\footnote{iorgov@bitp.kiev.ua}, O.~Lisovyy$^b$\footnote{lisovyi@lmpt.univ-tours.fr}} 
\address{ \it$^a$ Bogolyubov Institute for Theoretical Physics,  03143 Kyiv, Ukraine
 \vspace{0.2cm}\\
 $^b$ Institut Denis-Poisson, Universit\'e de Tours, Universit\'e d'Orl\'eans, CNRS, \\ Parc de Grandmont,
 37200 Tours, France \vspace{0.2cm}\\
    $^c$ Kyiv Academic University,  36 Vernadsky blvd., 03142 Kyiv, Ukraine \vspace{0.2cm}\\
 $^d$ Max-Planck-Institut f\"ur Mathematik, 53111 Bonn, Germany \vspace{0.2cm}\\
 $^e$ Center for Advanced Studies, Skolkovo Institute of Science and Technology, \\ 143026 Moscow, Russia
  \vspace{0.2cm}\\
 $^f$ National Research University Higher School of Economics, Department of
 Mathematics and International Laboratory of Representation Theory and Mathematical
 Physics, \\ 119048 Moscow, Russia\vspace{0.1cm}}

\begin{abstract}
We construct the general solution of a class of Fuchsian systems of rank $N$ as well as the associated isomonodromic tau functions in terms of semi-degenerate conformal blocks of $W_N$-algebra with central charge $c=N-1$. The simplest example is given by the tau function of the Fuji-Suzuki-Tsuda system, expressed as a Fourier transform of the 4-point conformal block with respect to intermediate weight. Along the way, we generalize the result of Bowcock and Watts on the minimal set of matrix elements of vertex operators of the $W_N$-algebra for generic central charge and prove several properties of semi-degenerate vertex operators and conformal blocks for $c=N-1$.
\end{abstract}

\section{Introduction}
The theory of monodromy preserving deformations has recently gained new insights from its connections to the two-dimensional conformal field theory. The most relevant for the present work is the solution of the inverse monodromy problem for rank 2 Fuchsian systems and  associated isomonodromic tau function  with Fourier transforms of $c=1$ Virasoro conformal blocks \cite{ILTe}. The simplest instance of this correspondence expresses the tau function of the Painlev\'e~VI equation in terms of 4-point conformal blocks  \cite{GIL1}; see also \cite{BSh} for a different proof of this statement.

The last relation has been extended to a number of confluent limits including Painlev\'e~V, IV and III \cite{GIL2,Nagoya}, as well as to the $q$-difference \cite{BSh2,JNS} and non-commutative \cite{BGM2} setting. On the other hand, the case of Fuchsian systems of higher rank $N>2$ is as yet only superficially explored. A natural guess  is that their fundamental solutions and tau functions are related to conformal blocks of $W_N$-algebras \cite{Gavrylenko} with integer central charge $c=N-1$ \cite{GIL1}. However, trying to make this claim well-defined and constructive one faces a number of obstacles.

The $W_N=W(\mathfrak{sl}_N)$ algebras appeared in \cite{FLuk2,FZ,FLuk1,ZamW} as extensions of the Virasoro algebra including chiral currents with higher 
spins. 
One of the problems appearing in their investigation is that, for general $N$, they do not have known explicit and convenient definition in terms of generators and relations.
An additional complication is that  $W_N$-algebras are not Lie algebras for $N>2$. Even though in the $W_3$ case it is possible to proceed in the study of the algebra and its representations (see for example \cite{W3,BS}), for $W_N$-algebras with $N>3$ the direct approach becomes quite involved.
The first definition of $W_N$-algebras for $N>3$ was given in terms of bosonic fields by means of the quantum Miura transformation 
\cite{FLuk2,FLuk1}.
Subsequently, new approaches appeared, the most productive of them being the quantum Drinfeld-Sokolov reduction \cite{FF}. It gives rise to the most general definition of $W$-algebras associating them 
with pairs formed by a simple Lie algebra and a nilpotent element therein \cite{KRW}. The recent progress in the representation theory of $W$-algebras \cite{Arakawa,deSoleKac} was also made in this framework. The present paper uses a bosonic realization of $W_N$-algebras coming from the quantum Miura transformation, but we also need some results from \cite{Arakawa}.

 A more significant problem is that, in contrast to the Virasoro  ($N=2$) case, the local $W$-invariance does not fix vertex operators uniquely: their descendant matrix elements cannot be reduced to the primary one.  This produces an infinite number of free parameters (even in the 3-point conformal blocks!) which have to be determined or constrained by other means such as crossing symmetry. For $N=3$, a characterization of the minimal set of independent  matrix elements was given in~\cite{BW}.  
 One may also adopt the point of view that such matrix elements can be fixed arbitrarily. However, the resulting vertex operators may be plagued by divergencies in the multi-point conformal blocks.
 Even their compositions with the fundamental degenerate vertex operator\footnote{Vertex operator corresponding to the highest weight vector of an irreducible representation of $W_N$-algebras associated to one of the $N$-dimensional fundamental representations of $\mathfrak{sl}_N$.} may give rise to a divergent series. A satisfactory definition of the vertex operator should produce a convergent expansion for the appropriate 4-point conformal block with consistent global analytic (monodromy) properties with respect to the position of the degenerate field. 
 
 In the case $c=N-1$, where the CFT/isomonodromy correspondence is expected to hold true, a definition of the general vertex operator for the $W_N$-algebra was proposed in \cite{Gavrylenko} by employing the isomonodromic tau functions and the corresponding 3-point Fuchsian systems. The elements of the basis of vertex operators are labeled in this approach by a finite number of moduli parameterizing the monodromy data \cite{GM16,Teschner17,GL16,CGT}. For generic central charge, analogous definition is not available so far, but it is expected to be consistent with an action of the algebra of Verlinde loop operators on the space of 3-point conformal blocks, see recent work \cite{CPT17}.

 There is a special class of vertex operators of $W_N$-algebras relevant to the construction in this paper. It consists of the so-called semi-degenerate vertex operators corresponding to 3-point conformal blocks with one field having highest weight labeled by $a \bs h_1$, where $a\in \mathbb{C}$ and $ \bs h_1$ is the highest weight of the first fundamental representation of $\mathfrak{sl}_N$. It is expected that arbitrary 3-point conformal blocks involving one semi-degenerate field (primary or its descendant) are uniquely determined by the primary ones as in the Virasoro case (for $N=3$, this statement \cite{KMS} can be deduced from the results of \cite{BW} quoted above). The multi-point conformal blocks of $W_N$-algebras which appear in the AGT-type relation \cite{AGT,Wyl,FL3,MM} to Nekrasov instanton partition functions \cite{Nekr} are precisely those obtained by compositions of semi-degenerate vertex operators.
 
 On the  differential equations side, the tau function of a Fuchsian system with $n$ regular singular points can be written \cite{GL16,CGL} as a Fredholm determinant whose integral kernel is expressed in terms of fundamental solutions of $3$-point Fuchsian systems which arise upon decomposition of the $n$-punctured Riemann sphere into pairs of pants. While for $N=2$ these auxiliary systems can be explicitly solved in terms of Gauss hypergeometric functions, the construction of $3$-point solutions (series or integral representations, connection formulas, etc) for $N\ge 3$ is a major open problem. 
 
  An important exception is given by the 3-point Fuchsian systems with one of the singular points having special spectral type $\lb N-1,1\rb$. Their solutions can be expressed in terms of the Clausen-Thomae hypergeometric functions $_NF_{N-1}$. For this reason, the constructions of \cite{GL16,CGL} can be made completely explicit for Fuchsian systems with 2 generic singular points and the remaining $n-2$ singularities of the above special type. In this paper, such systems are called \textit{semi-degenerate};
 they are the closest higher rank relatives of the $N=2$ ones.
   The specialization of the Schlesinger isomonodromic deformation equations to semi-degenerate $4$-point Fuchsian  system (i.e.  higher rank semi-degenerate analog of Painlev\'e~VI) is known as the Fuji-Suzuki-Tsuda system. It was discovered in \cite{FujiSuzuki,Suzuki} as a similarity reduction of Drinfeld-Sokolov hierarchies, and related to deformations of Fuchsian equations in
 \cite{Tsuda}, where semi-degenerate Fuchsian systems with arbitrary $n\ge 4$ also appear.
 
 The aim of this paper is to extend the results of \cite{ILTe} to higher rank case by relating fundamental solutions and tau functions of semi-degenerate Fuchsian systems to Fourier transformed semi-degenerate conformal blocks of the $W_N$-algebras with $c=N-1$. In particular, we claim that, under suitable normalization of vertex operators,
 \be\label{claim}
 \tau\lb \bs z\rb = 
 \sum_{\boldsymbol{w}_1,\ldots,\boldsymbol{w}_{n-3}\in \mathfrak R} 
  e^{(\boldsymbol{\beta}_1,\boldsymbol{w}_1)+\ldots+(\boldsymbol{\beta}_{n-3},\boldsymbol{w}_{n-3})}{\hspace{-0.5cm}
  \begin{tikzpicture}[baseline,yshift=-0.5cm,scale=1.3]
  \draw [thick] (0.2,0) -- (2.7,0);
  \draw [thick,dashed] (1,0) -- (1,0.8);
  \draw [thick,dashed] (2.2,0) -- (2.2,0.8);
  \draw (3,0) node {$\ldots$};  
  \draw [thick] (3.3,0) -- (5.6,0);
  \draw [thick,dashed] (3.8,0) -- (3.8,0.8);
  \draw [thick,dashed] (4.8,0) -- (4.8,0.8);
  \draw (0.2,0) node[left] {$\infty$};
  \draw (5.6,0) node[right] {$0$};
  \draw (0.6,0) node[below] {\scriptsize $-\bs{\theta}_{n-1}$};
  \draw (5.3,0) node[below] {\scriptsize $\bs{\theta}_{0}$};
  \draw (1.8,0) node[below] {\scriptsize $\bs{\sigma}_{n-3}+\bs{w}_{n-3}$};
  \draw (3,0) node {$\ldots$};
  \draw (4.35,0) node[below] {\scriptsize $\bs{\sigma}_{1}+\bs{w}_{1}$};
  \draw (1,0.8) node[above] {$z_{n-2}$};   
  \draw (2.2,0.8) node[above] {$z_{n-3}$};  
  \draw (3.8,0.8) node[above] {$z_{2}$}; 
  \draw (4.8,0.8) node[above] {$z_{1}$}; 
  \draw (1,0.4) node[right] {\scriptsize $a_{n-2}$};    
  \draw (2.2,0.4) node[right] {\scriptsize $a_{n-3}$};  
  \draw (3.8,0.4) node[right] {\scriptsize $a_{2}$};   
  \draw (4.8,0.4) node[right] {\scriptsize $a_{1}$};        
  \end{tikzpicture}}
 \ee
 where $\tau\lb \boldsymbol{z}\rb$ is the isomonodromic tau function  depending on the positions $\bs z=\left\{z_1,\ldots,z_{n-2}\right\}$ of the special punctures, the generic punctures are located at $z_0=0$ and $z_{n-1}=\infty$, and the trivalent graph on the right denotes appropriate $n$-point conformal block. The parameters $\bs\theta_0,-\bs \theta_{n-1}\in\mathbb C^{N}$, $a_1,\ldots,a_{n-2}\in\mathbb C$ assigned to external edges describe local monodromy exponents of the Fuchsian system. The labels $\bs\sigma_1,\ldots,\bs\sigma_{n-3}\in\mathbb C^{N}$ of the internal edges as well as Fourier momenta $\bs\beta_1,\ldots,\bs\beta_{n-3}\in\mathbb C^{N}$ are explicitly related to the remaining moduli of semi-degenerate monodromy; the components of all vectors in $\mathbb C^{N}$  sum up to zero. The summation in \eqref{claim} is carried over the root lattice $\mathfrak R$ of $\mathfrak{sl}_N$. Analogous statement for the fundamental solution involves extra degenerate insertions.
 
 In the way of establishing the correspondence, we proved several statements about conformal blocks of the $W_N$-algebras which, to the best of our knowledge, remained so far at the level of folklore for $N>3$. For arbitrary central charge, we developed a reduction procedure of matrix elements of vertex operators to a minimal set, thereby generalizing the results of \cite{BW}.
 For $c=N-1$, we proved that the descendant 3-point functions involving semi-degenerate field are uniquely expressed in terms of the 3-point function of primaries.  
 For the fundamental degenerate field, we found restrictions (fusion rules) to be satisfied to allow for non-vanishing 3-point functions.
 They are then used to prove the well-known hypergeometric formulas \cite{FL1} for the 4-point conformal blocks with one semi-degenerate and one degenerate field using the rigidity property of the associated Fuchsian system.
 
    We end this introduction by mentioning a few more relevant papers. 
    A survey of recent results in the representation theory of $W$-algebras may be found in \cite{ArakawaLec}. The properties of semi-degenerate $W_N$-conformal blocks describing correlation functions of the Toda CFT and their gauge theory counterparts have been studied, for instance, in \cite{GF,Bul}. An interesting direction which is in a sense close to ours is the construction of integral representation of the 4-point conformal blocks of $W_3$-algebra involving one semi-degenerate field of higher level (whose matrix elements cannot be reduced to primary ones only) and one fundamental degenerate field \cite{BHS}. This construction is based on the middle convolution from the Katz theory of rigid systems. 
Connections between Fuchsian systems and $W$-algebras were also studied in \cite{BER}, with further links to topological recursion suggested in \cite{BE}.
 
 The paper is organized as follows. In Section \ref{sec2}, we introduce semi-degenerate Fuchsian systems and provide an explicit parameterization of their monodromy by suitable coordinates (Proposition \ref{propparam}). Section~\ref{sec3} is devoted to $W_N$-algebras and their representations. The minimal set of matrix elements is described by Theorem~\ref{thmgenME}, after which we proceed to the proof of uniqueness of semi-degenerate vertex operator  (Proposition~\ref{propsdvo}) and fusion rules for completely degenerate fields. Section~\ref{sec4} computes the operator-valued monodromy of semi-degenerate conformal blocks with respect to positions of additional degenerate fields. Diagonalizing this monodromy by Fourier transform, in Section~\ref{sec5} we obtain the fundamental solution (Theorems~\ref{solinfty} and~\ref{soly0}) and the tau function (Proposition~\ref{proptaucb}) of the semi-degenerate Fuchsian systems in terms of $W_N$-algebra conformal blocks. Some technical results are relegated to appendices. \vspace{0.2cm}

{\small\textbf{Acknowledgements}. We would like to thank M. Bershtein, A. Marshakov, R. Santachiara and G. Watts for useful discussions and comments. The present work was supported by the CNRS/PICS project ``Isomonodromic deformations and conformal field theory''. The work of P.G. was partially supported the Russian Academic Excellence Project `5-100' and by the RSF grant No. 16-11-10160. In particular, odd-numbered formulas of Section 3 have been obtained using support of Russian Science Foundation. P.G. is a Young Russian Mathematics award winner and would like to thank its sponsors and jury. N.I. thanks Max Planck Institute for Mathematics (Bonn), where a part of this research was done, for hospitality and
excellent working conditions.}

\section{Semi-degenerate Fuchsian systems\label{sec2}}

\subsection{Generalities}

We are interested in the analysis of Fuchsian systems of rank $N$ having $n$ regular singular points 
$\bs z:=\left\{z_0,\ldots,z_{n-2},z_{n-1}\equiv\infty\right\}$ on the Riemann sphere $\mathbb{CP}^1$:
\be\label{FS}
\frac{d\Phi\lb y\rb}{dy} = \Phi\lb y\rb A\lb y\rb,\qquad A\lb y\rb=\sum_{k=0}^{n-2} \frac{A_k}{y-z_k}\,.
\ee
Here the residues $A_0,\ldots,A_{n-2}\in\operatorname{Mat}_{N\times N}\left(\mathbb C\right)$ and
$\Phi\lb y\rb$ is the fundamental $N\times N$ matrix solution which may  be normalized as
$\Phi\lb y_0\rb = \mathbb{I}$, with $y_0\in \mathbb{CP}^1\backslash \bs z$. It will be assumed that the matrices $A_0,\ldots,A_{n-2}$ and $A_{n-1}:=-\sum_{k=0}^{n-2} A_k$ are diagonalizable and non-resonant, i.e. the pairwise differences of the eigenvalues of each $A_k$ are not non-zero integers.  

The solution $\Phi\lb y\rb$ is a multivalued function on $\mathbb{CP}^1\backslash \bs z$. Its monodromy realizes a representation of the fundamental group $\pi_1\lb \mathbb{CP}^1\backslash \bs  z, y_0\rb$ in $GL\lb N,\mathbb{C}\rb$. 
This group  is generated by the paths $\xi_0,\ldots,\xi_{n-1}$ around the points 
$z_0,\ldots,z_{n-1}$ on $\mathbb{CP}^1$ indicated in Fig.~\ref{fig1}, which satisfy one relation $\xi_0\cdots \xi_{n-1}=1$. In what follows, their orientations will be referred to as positive.
Denoting by $M_k$ the monodromy of $\Phi\lb y\rb$ along the loop $\xi_k$, we similarly have $M_0 \cdots M_{n-1}=\mathbb{I}$. 

\begin{figure}[h!]
              \centering
\begin{tikzpicture}[baseline,yshift=-0.5cm,scale=0.7]
\draw (-1,0) .. controls +(0,-1) and +(0,-1) .. (1,0) .. controls +(0,1) and +(-3,2.5) .. (7,0);
\draw[->,>=latex] (7,0) .. controls +(-3,3) and +(0,1.5) .. (-1,0) ;
\draw (3,0) .. controls +(0,-1) and +(0,-1) .. (5,0) .. controls +(0,0.5) and +(-2,1) .. (7,0);
\draw[->,>=latex] (7,0) .. controls +(-3,2) and +(0,0.7) .. (3,0) ; 
\draw (7,0) .. controls +(-2,3) and +(0,3) .. (-3,0); 
\draw[->,>=latex] (7,0) .. controls +(-2,-3) and +(0,-3) .. (-3,0); 
\draw[fill] (0,0) circle (0.07);
\draw[fill] (4,0) circle (0.07);
\draw (0,0) node[above] {$z_0$}; 
\draw (4,0) node[above] {$z_{n-2}$}; 
\draw (2,0) node {$\ldots$}; 
\draw (7,0) node[right] {$y_0$};
\draw (-3,0) node[right] {$\xi_{n-1}$};
\draw (-0.9,-0.9) node {$\xi_{0}$};
\draw (3.1,-0.9) node {$\xi_{n-2}$};
\end{tikzpicture}      
\caption{Basis of loops $\xi_0,\ldots,\xi_{n-1}$ in $\pi_1\lb \mathbb{CP}^1\backslash \bs  z, y_0\rb$.\label{fig1}}
\end{figure}

The fundamental matrix $\Phi\lb y\rb$ is uniquely fixed by the following properties:
\begin{enumerate}
\item[(a)] $\Phi\lb y\rb$ is holomorphic and invertible on  the universal cover of $\mathbb{CP}^1\backslash \bs z$ and has constant monodromy under analytic continuation.
\item[(b)] $\Phi\lb y\rb$ satisfies the normalization condition $\Phi\lb y_0\rb = \mathbb{I}$.
\item[(c)]
 In sufficiently small neighborhoods of $z_k$, $k=0,\ldots,n-1$,  the behavior of $\Phi\lb y\rb$ is 
\be
\Phi\lb y\to z_k\rb=C_k \lb z_k-y\rb^{\Theta_k} G_k\lb y\rb, 
\ee 
where $G_k\lb y\rb$ is holomorphic and invertible in the vicinity  $z_k$; $C_k$ is a non-degenerate constant matrix;
$\Theta_k$ is a diagonal matrix conjugate to $A_k$. 
(The asymptotics at $z_{n-1}=\infty$ should be rewritten in terms of a suitable local parameter). Note that $M_k=C_k e^{2\pi i \Theta_k} C_k^{-1}$.
\end{enumerate}

\begin{defn}
The Riemann-Hilbert problem associated with the Fuchsian system (\ref{FS}) is the problem of reconstruction of $\Phi\lb y\rb$ satisfying the conditions {(a)--(c)} for a given monodromy data: $M_k\in GL(N,\mathbb{C})$, $k=0,\ldots,n-1$, subject to the relation
$M_0 \cdots M_{n-1}=\mathbb{I}$. 
\end{defn}

Instead of normalizing the fundamental solution $\Phi(y)$ by the condition (b), we could also use another normalization 
which combines (b) and (c): namely, one may fix the connection matrix $C_l=\mathbb{I}$ at one of the singular points.

Different choices of the normalization point $y_0$ lead to an overall conjugation of all monodromies. 
We identify the corresponding monodromy data and consider the space
\be
\mathcal{M}=\mathrm{Hom}\lb\pi_1\lb\mathbb{CP}^1\backslash \bs z,y_0\rb, GL\lb N,\mathbb{C}\rb\rb/ GL\lb N,\mathbb{C}\rb\,.
\ee
It is often convenient to work with the slice $\mathcal M_{\Theta}\subset\mathcal M$ corresponding to fixed local monodromy exponents $\Theta=\left\{\Theta_0,\ldots,\Theta_{n-1}\right\}$.

Besides $\Phi\lb y\rb$, we will be also interested in the isomonodromic tau-function $\tau\lb \bs z\rb$ of Jimbo-Miwa-Ueno \cite{JMU}. 
It is defined by the following 1-form
\be\label{isomtau}
d_{\bs z} \log \tau\lb \bs z\rb := \frac 1 2  \sum_{k=0}^{n-1} \mathrm{res}_{y=z_k} \mathrm{Tr}\, A^2\lb y\rb dz_k,
\ee
which is closed provided the monodromy of \eqref{FS} is kept constant.
The tau function $\tau\lb \bs z\rb\equiv\tau\lb \bs z\left|\Theta,\mathbf m\right.\rb$ with $\mathbf m\in\mathcal M_{\Theta}$ is a generating function of the Hamiltonians which govern the isomonodromic evolution of $A_0,\ldots,A_{n-2}$ with respect to times $\bs z$.

\subsection{Semi-degenerate monodromy}
Let $\mathbb Y$ be the set of partitions $\lambda=\left[\lambda_1,\ldots \lambda_{\ell}\right]$, $\lambda_1\ge\ldots\ge\lambda_{\ell}>0$, and $\mathbb Y_k$ be the set of all partitions of $k\in\mathbb Z_{\ge0}$. One can decompose the space of Fuchsian systems \eqref{FS} according to their spectral type $\boldsymbol s=( s^{(0)},\ldots,s^{(n-1)})\in\mathbb Y_N^n$, where the partition $s^{(i)}\vdash N$ encodes the multiplicities of the eigenvalues of $\Theta_i$ or $A_i$. Thus, for example, $\ell(s^{(i)})$ is the number of distinct eigenvalues of~$\Theta_i$ and $s^{(i)}_1$ is the multiplicity of its most degenerate eigenvalue. 

The dimension of the space $\mathcal M_{\Theta}$ of monodromy data for irreducible systems of spectral type $\mathbf s$ coincides with the number of accessory parameters, and is known to be given by
\be
\operatorname{dim}\mathcal M_{\Theta}=\lb n-2\rb N^2+2-\sum_{i=0}^{n-1}\sum_{j=1}^{\ell_i}\lb s^{(i)}_j\rb^2\,.
\ee
Generic Fuchsian systems have spectral type $\mathbf s_{\mathrm{gen}}=\lb\lb 1^N\rb,\ldots ,\lb 1^N\rb\rb$. It then follows from the last formula that
\be\label{dimgen}
\operatorname{dim}\mathcal M_{\Theta,\mathrm{gen}} =2\lb n-3\rb \lb N-1\rb+\lb n-2\rb\lb N-1\rb\lb N-2\rb.
\ee
This expression has a geometric interpretation. The $n$-punctured Riemann sphere can be decomposed into $n-2$ pairs of pants (3-punctured spheres) by $n-3$ closed curves. To each of these curves one may assign $2\lb N-1\rb$ monodromy parameters which play the role of Fenchel-Nielsen-type coordinates (lengths and twists) and give the 1st term in \eqref{dimgen}. The 2nd term comes from $\lb N-1\rb\lb N-2\rb$ parameters associated to each pair of pants with fixed conjugacy classes of local monodromies at 3 boundaries. The presence of such parameters is the principal new feature of the higher rank $N\ge 3$.

We are interested in the Fuchsian systems with 2 generic $\lb 1^N\rb$-punctures at $z_0$ and $z_{n-1}$,  and $n-2$ singular points  of spectral type $\lb N-1,1\rb$ at $z_1,\ldots,z_{n-2}$.
The systems of this type will be called \textit{semi-degenerate}. The dimension of the relevant space of monodromy data is readily computed to be
\be\label{dimsemideg}
\operatorname{dim}\mathcal M_{\Theta,\mathrm{s-d}} =2\lb n-3\rb\lb N-1\rb.
\ee
For $n=3$, this dimension vanishes, meaning that the Fuchsian system with 2 generic punctures and one puncture of type $\lb N-1,1\rb$ is rigid. The conjugacy class of monodromy is then completely determined by the local exponents~$\Theta$, i.e. the pants carry no internal moduli. For $n\ge4$, there exist decompositions of the $n$-punctured sphere into such semi-degenerate pants, which explains the difference between (\ref{dimgen}) and \eqref{dimsemideg}.

Our next task is to provide an explicit parameterization of semi-degenerate monodromy. The construction of solution of the corresponding Riemann-Hilbert problem  and the associated isomonodromic tau function constitutes the main goal of the present work.

\begin{ass}\label{assmon}
The monodromy matrices $M_k\in SL\lb N,\mathbb{C}\rb$, $k=0,\ldots,n-1$  satisfying the cyclic condition $M_0\cdots M_{n-1}=\mathbb I$ are assumed to be diagonalizable, i.e.
 $M_k\sim \exp(2\pi i \Theta_k)$, where $\Theta_k=\operatorname{diag}\ \bs\theta_k$ with $\bs\theta_k=(\theta_k^{(1)},\ldots,\theta_k^{(N)})\in\mathbb{C}^N$ are traceless diagonal matrices. For $k=1,\ldots,n-2$, these matrices are fixed to be
\be
\bs \theta_k =  a_k \lb\tfrac{N-1}{N},-\tfrac 1N,\ldots,-\tfrac 1N\rb,\qquad a_k\in\mathbb{C}.  \ee
It is further assumed that the products
$M_{[k]}:=M_0\cdots M_k$ with $k=0,\ldots,n-2$  are also diagonalizable and  their eigenvalues $\mathrm{Spec}\, M_{[k]}$ are pairwise distinct:
\be M_{[k]}\sim \exp\lb 2\pi i \fs_k\rb,\qquad\fs_k=\operatorname{diag} \bs\sigma_k,\qquad\bs\sigma_k=(\sigma_k^{(1)},\ldots,
\sigma_k^{(N)}),
\ee
where $\operatorname{Tr}\fs_k=0$. Note that $M_{[0]}=M_0$, $M_{[n-2]}=M_{n-1}^{-1}$, so that we can  identify $\fs_0=\Theta_0$, $\fs_{n-2}=-\Theta_{n-1}$, $\bs\sigma_0=\bs\theta_0$, $\bs\sigma_{n-2}=-\bs\theta_{n-1}$.
\end{ass}

For $n=3$, the semi-degenerate monodromy is described by the following result, see e.g. \cite{Beukers}. 

\begin{lem}[Rigidity Lemma]\label{KatzLemma}
If $M_A, M_B\in GL(N,\mathbb{C})$ are diagonalizable with non-intersecting sets of eigenvalues 
$\mathrm{Spec}\,M_A=\{\al_1,\ldots,\al_N\}$,
$\mathrm{Spec}\,M_B=\{\beta_1,\ldots,\beta_N\}$ and $M_B^{-1} M_A$ is  a reflection (a rank 1 perturbation of the identity matrix) then 
there exists a unique (up to overall rescaling) basis in which 
  \[
  M_A=\left(\begin{array}{ccccc}
   0 & 0 & 0 & \ldots & (-1)^{N+1} e_N(A) \\
   1 & 0 & 0 & \ldots & (-1)^{N} e_{N-1}(A) \\
   0 & 1 & 0 & \ldots & (-1)^{N-1} e_{N-2}(A) \\
   \cdot & \cdot & \cdot & \ddots & \cdot \\ 
   0 & 0 & \ldots & 1 & e_1(A)
   \end{array}\right),\quad 
   M_B=\left(\begin{array}{ccccc}
      0 & 0 & 0 & \ldots & (-1)^{N+1} e_N(B)\\
      1 & 0 & 0 & \ldots & (-1)^{N} e_{N-1}(B) \\
      0 & 1 & 0 & \ldots & (-1)^{N-1} e_{N-2}(B) \\
      \cdot & \cdot & \cdot & \ddots & \cdot \\ 
      0 & 0 & \ldots & 1 & e_1(B)
      \end{array}\right),
  \]
where $e_k\lb A\rb$ and $e_k\lb B\rb$ denote the $k$-th elementary symmetric polynomials in the eigenvalues of $M_A$ and $M_B$, respectively.
In this case, $\mathrm{Spec}\,M_B^{-1} M_A =\bigl\{\prod_{k=1}^N\al_k \beta_k^{-1},1,\ldots,1\bigr\}$.   
\end{lem}

The matrices $W_A$ and $W_B$ defined by 
$\lb W_A\rb_{kl}=\al_k^{l-1}$ and $\lb W_B\rb_{kl}=\beta_k^{l-1}$ diagonalize, respectively, $M_A$ and $M_B$:
\[
 W_A M_A W_A^{-1} =D_A, \qquad W_B M_B W_B^{-1} =D_B\,,
\]
where $D_A=\mathrm{diag} \lb \al_1,\ldots,\al_N\rb$ and  $D_B=\mathrm{diag} \lb \beta_1,\ldots,\beta_N\rb$.
The matrix $W_B W_A^{-1}$ relates the eigenvectors of $M_A$ and $M_B$. Its matrix elements are
\be\label{WBWA}
(W_B W_A^{-1})_{kl}=\prod_{s(\ne l)} \frac{ \beta_k-\al_s} {\al_l-\al_s}\,.
\ee
Note that the general form of a matrix which relates a basis where $M_A$ is diagonal to another basis where $M_B$ is diagonal is given by 
$R_B W_B W_A^{-1} R^{-1}_A$, where $R_A$ and $R_B$ are non-degenerate diagonal matrices.

Now one may use Lemma~\ref{KatzLemma} to parameterize recursively the  monodromy matrices $M_k$ of semi-degenerate Fuchsian systems. 
To this end observe that it suffices to parameterize instead a related set of matrices
$M_{[k]}=M_0\cdots M_k$ with $k=0,\ldots,n-2$. Indeed, we have $M_{0}=M_{[0]}$, $M_{n-1}=M_{[n-2]}^{-1}$ and  $M_{k}=M_{[k-1]}^{-1}M_{[k]}$ for $k=1,\ldots,n-2$.
\begin{prop}\label{propparam}
Let $M_k\in SL(N,\mathbb{C})$, $k=0,\ldots,n-1$ be the monodromy matrices of a semi-degenerate Fuchsian system satisfying genericity conditions of Assumption~\ref{assmon}. They can be parameterized as follows:
\be\label{MkWk}
M_{[k]}=W_{[k]}^{-1} \exp\lb2\pi i \fs_k\rb W_{[k]}\,,
\ee
\be\label{WWW}
W_{[k]} = R_{k} W_{k+1} R_{k+1} \cdots R_{n-3} W_{n-2} R_{n-2}\,,
\ee
where $R_k=\mathrm{diag}\lb\bs r_k\rb$ are diagonal matrices from $SL(N,\mathbb{C})$ and 
\be\label{Wm}
\lb W_m\rb_{kl}=\prod_{s(\ne l)} \frac{ e^{2\pi i (\sigma_{m-1}^{(k)}-a_m/N)}-e^{2\pi i \sigma_{m}^{(s)}}} {e^{2\pi i \sigma_{m}^{(l)}}-e^{2\pi i \sigma_{m}^{(s)}}}\,.
\ee
\end{prop}

\begin{proof}
The idea is to use Lemma~\ref{KatzLemma} successively for the pairs of matrices 
\[
M_A =e^{2\pi i a_k/N} M_{[k]},\qquad M_B=M_{[k-1]}, \qquad k=n-2,\ldots,1,
\]
where the factor $e^{2\pi i a_k/N}$ ensures that the eigenvalue of $M_B^{-1}M_A=e^{2\pi i a_k/N}M_k$ with multiplicity $N-1$ is equal to 1.
We start the parameterization from the pair 
$$M_A=e^{2\pi i a_{n-2}/N} M_{[n-2]}, \qquad M_B=M_{[n-3]},$$ 
assuming that $M_A$ is diagonal:
$M_A=e^{2\pi i a_{n-2}/N} \exp\lb 2\pi i \fs_{n-2}\rb$, where $\fs_{n-2} = - \Theta_{n-1}$. Then 
\[
M_B=R_{n-2}^{-1} W_{n-2}^{-1} R_{n-3}^{-1} \exp\lb 2\pi i \fs_{n-3}\rb R_{n-3} W_{n-2} R_{n-2}\,,
\]
where, as follows from (\ref{WBWA}), $W_{n-2}$ is given by (\ref{Wm})
and $R_k=\mathrm{diag}\lb \bs r_k\rb$ are arbitrary diagonal matrices from $SL\lb N,\mathbb{C}\rb$.
Continuing the recursive procedure, we get the parameterization (\ref{MkWk}) for all $M_{[k]}$.
The matrix $W_{[k]}$ defined by (\ref{WWW}) relates the bases which diagonalize $M_{[n-2]}=M_{n-1}^{-1}$ and $M_{[k]}$.
\end{proof}
Observe that the diagonal matrix $R_{k}$ cancels out in (\ref{MkWk}); however, we keep it for later use. Since  $R_{n-2}$ only produces an overall conjugation of all monodromies, it does not enter into the parameterization of  $\mathcal{M}_{\Theta,\mathrm{s-d}}$.
Thus the semi-degenerate monodromy is parameterized by  $\bs \sigma_k$, $\bs r_k$ with $k=1,\ldots,n-3$, which of course agrees with the dimension \eqref{dimsemideg}.

\subsection{Three-point case}

This subsection gives an explicit solution of the semi-degenerate Fuchsian system (\ref{FS}) with 3 singular points $z_0=0$, $z_1=1$ $z_2=\infty$, and
the connection $A\lb y\rb$ having {\em traceless} residues $A_0$, $A_1$, $A_\infty=-A_0-A_1$ at these poles.
We suppose that $A_0$, $A_1$, $A_\infty$ are diagonalizable to $\Theta_0$, $\Theta_1$, $\Theta_\infty$. Moreover, it is convenient to choose the gauge so that $A_\infty$ is diagonal.
Thus 
\be\label{FS3pA}
A_0 = G_0^{-1} \Theta_0 G_0,\qquad A_1=G_1^{-1} \Theta_1 G_1, \qquad A_\infty =  \Theta_\infty,
\ee
where $G_0$, $G_1\in GL\lb N,\mathbb{C}\rb$ and
\be\label{FS3pTh}
\begin{gathered}
\Theta_0=\mathrm{diag}\lb \theta^{(1)}_0,\ldots,\theta^{(N)}_0\rb,\quad
\Theta_\infty=\mathrm{diag}\lb\theta^{(1)}_\infty,\ldots,\theta^{(N)}_\infty\rb,\quad 
\operatorname{Tr} \Theta_0=\operatorname{Tr} \Theta_\infty=0,\\
\Theta_1=a \cdot \operatorname{diag}\lb\tfrac{N-1}{N},-\tfrac 1N,\ldots,-\tfrac 1N\rb.
\end{gathered}
\ee
Recall that while $\Theta_1$ has an eigenvalue of multiplicity $N-1$, all eigenvalues of  $\Theta_0$ and $\Theta_\infty$ are distinct.
Such data correspond to a rigid local system and the matrix elements of $A_1$ can be derived from an additive variant of Lemma~\ref{KatzLemma}:
\be\label{ME_Ajm}
 \lb A_1\rb_{jm}=-\frac{r_j}{r_m}\cdot
 \frac{\prod_{k}(\theta^{(j)}_\infty-a/N+\theta^{(k)}_0)}{\prod_{k(\ne m)}(\theta^{(m)}_\infty-\theta^{(k)}_\infty)}
 -\delta_{jm}\frac{a}{N}\,,
\ee
where $r_1,\ldots,r_N$ are arbitrary non-zero parameters. They
appear due to the possibility of overall conjugation of $A_0$, $A_1$, $A_\infty$ by 
the diagonal matrix $R=\mathrm{diag}\lb r_1,\ldots,r_N\rb$,  preserving the diagonal form of $A_\infty$.

\begin{theo}
The solution of the Fuchsian system 
\be\label{FS3p}
\frac{d\Phi\lb y\rb}{dy} = \Phi\lb y\rb A\lb y\rb,\qquad A\lb y\rb =\frac{A_0}{y}+\frac{A_1}{y-1},
\ee
with $A_1$ fixed by (\ref{ME_Ajm}) and $A_0=-A_1-\Theta_{\infty}$, which has the asymptotics $\Phi\lb y\rb=y^{-\Theta_\infty}\lb \mathbb I+O\lb y^{-1}\rb\rb$ as $y\to\infty$  is 
\[
 \Phi_{jm}\lb y\rb= N_{jm} y^{-\theta^{(j)}_\infty-1+\delta_{jm}} 
 \lb 1-\tfrac 1y\rb^{-a/N} \times\qquad\qquad\quad
 \]\be\label{FS3psol}
 \times
{}_NF_{N-1}\left(\genfrac{}{}{0pt}{0}{\bigl\{1-\delta_{jm}-a/N+\theta^{(k)}_0+\theta^{(j)}_\infty\bigr\}_{k=\overline{1,N}} } 
{\bigl\{1+\theta^{(j)}_\infty-\theta^{(k)}_\infty
+\delta_{mk}-\delta_{jm}\bigr\}_{k=\overline{1,N};k\ne j}}\Bigg|\,\frac{1}{y}\right),
\quad j,m=1,\ldots,N,
\ee
where 
\[
N_{jm}=\begin{cases}\displaystyle\frac{\lb A_1\rb_{jm}}{\theta^{(m)}_\infty-\theta^{(j)}_\infty-1}, 
& \quad j\ne m,\\
\qquad\;\; 1, & \quad j=m,
\end{cases}
\]
and $_NF_{N-1}\lb\ldots\left|\,x\right.\rb$ denotes the generalized hypergeometric function.
\end{theo}

Let us comment on the computation of the coefficients  $N_{jm}$.
They can be derived from the expansion of $\Phi\lb y\rb$ near  infinity. Indeed, as $y\to\infty$, one has
\[
 A\lb y\rb = -\frac{A_\infty}{y}+\frac{A_1}{y^2}+O\lb y^{-3}\rb.
\]
The solution $\Phi(y)$ can be found iteratively using the expansion
\be\label{FS3psol-exp}
 \Phi\lb y\rb= y^{-A_\infty} \lb \mathbb I+H_1y^{-1}+H_2 y^{-2}+O\lb y^{-3}\rb\rb.
\ee
Substituting this expression into the Fuchsian system yields 
$H_1=[H_1,A_\infty]-A_1$.
Under the non-resonance condition for $A_\infty$, it follows that
\[
 \lb H_1\rb_{jm}=\frac{\lb A_1\rb_{jm}}{\theta^{(m)}_\infty-\theta^{(j)}_\infty-1}.
\]
Comparing (\ref{FS3psol}) and  (\ref{FS3psol-exp}) as $y\to \infty$, one finds that
the coefficients of the off-diagonal leading terms and  diagonal  next-to-leading terms are given by
the matrix elements of $H_1$. In particular,  $N_{jm}=\lb H_1\rb_{jm}$ for $j\ne m$.

\section{$W_N$-algebras and their representations\label{sec3}}
It is important for us that  $W_N$-algebras with central charge $c=N-1$ have a bosonic realization related to the bosonic realization of 
    $\widehat{\mathfrak{sl}}_N$ at level~$1$.  Since this value of $c$ is exactly what we need for applications to isomonodromy, 
    the bosonic realization of $W_N$-algebras will be used as their basic definition. Thanks to the boson-fermion correspondence, in this case there also exists a fermionic realization of $W_N$-algebras.

\subsection{Definition of $W_N$-algebras with  $c=N-1$}
We are going to define the $W_N$-algebras, $W_N=W(\mathfrak{sl}_N)$, with  $c=N-1$, as abstract operator algebras starting from
their realizations in terms of $N$ free bosonic fields $\phi_k(z)$, $k=1,\ldots, N$, subject to one relation
$\sum_{k=1}^N \phi_k(z)=0$. The currents $J_k(z)=i\partial \phi_k(z)$  have the operator product expansions (OPE) of the form
\be\label{OPEJN}
 J_k(z) J_l(z')=\frac{\delta_{kl}-\frac{1}{N}}{\lb z-z'\rb^2}+\mathrm{regular}.
\ee
The modes $a^{(k) }_p$ of the currents $J_k(z)$ are defined by
\be\label{JNmodes}
 J_k(z)=\sum_{p\in \mathbb{Z}} \frac{a^{(k)}_p}{z^{p+1}}\,.
\ee
These currents define $W_N$-algebra currents $W^{(2)}\lb z\rb, \ldots, W^{(N)}\lb z\rb$ as 
sums of normal-ordered monomials:
\be\label{Wgen}
W^{(j)}\lb z\rb=\sum_{1\le i_1<\cdots<i_j\le N} :J_{i_1}\lb z\rb \cdots J_{i_j}\lb z\rb:\,, \qquad j=2,\ldots,N.
\ee
For example, 
\[
 W^{(2)}\lb z\rb\,=\sum_{k<l} :J_k\lb z\rb J_l\lb z\rb:=-\frac{1}{2} \sum_{k=1}^N :J_k\lb z\rb^2:=-T\lb z\rb,
\]
where $T\lb z\rb$ is the holomorphic component of the energy-momentum tensor.
The OPEs of currents (\ref{Wgen}) between themselves can be rewritten in terms of products of currents from the same set and their derivatives giving an
abstract definition of $W_N$-algebras. 
{}From the OPE of the energy-momentum tensor $T\lb z\rb$  with itself we get the central charge $c=N-1$.
The modes $W^{(j)}_p$ of the currents $W^{(j)}\lb z\rb$ are defined by 
\[
 W^{(j)}(z)=\sum_{p\in \mathbb{Z}} \frac{W^{(j)}_p}{z^{p+j}}.
\]

We will not need the explicit form of all OPEs of the currents $W^{(j)}\lb z\rb$  and the explicit formulas for the commutation relations of their modes $W^{(j)}_p$. However, it will be important for us that the commutation relations respect two structures on the $W_N$-algebra: $\mathbb{Z}$-gradation with respect to the adjoint action
of $L_0=-W_0^{(2)}$ 
\be\label{degL0}
\operatorname{deg}_{L_0}W^{(j)}_p=-p,
\ee
and quasi-commutativity with respect to the ${W}$-filtration \cite{Arakawa} defined by 
\be\label{degW}
\operatorname{deg}_{{W}}  W^{(j)}_p=j-1.
\ee
Namely, we will need the relations
\begin{align}\label{crL0}
\mathrm{deg}_{L_0}\,[W^{(j_1)}_{p_1},W^{(j_2)}_{p_2}]=&\,-(p_1+p_2),
\\ \label{crW}
\mathrm{deg}_{{W}} \,[W^{(j_1)}_{p_1},W^{(j_2)}_{p_2}]<&\; \lb j_1-1\rb+\lb j_2-1\rb.
\end{align}
The latter inequality means that the modes of $W_N$-algebra currents commute up to elements of smaller degree with respect
to $\mathrm{deg}_{{W}}$ (quasi-commutativity which is commutativity of the corresponding graded algebra).

\subsection{Vertex operators}
Given $\boldsymbol{\theta}=(\theta_1,\ldots,\theta_N)\in \mathbb{C}^N$ with $\sum_{k=1}^N \theta_k=0$, one may introduce the
exponential vertex operator
\[
 V_{\boldsymbol{\theta}}(z)=\,:e^{i \lb \boldsymbol{\theta},\phi(z)\rb}:
\]
It has the following OPEs with the currents $J_k\lb z\rb$:
\[
 J_k\lb z\rb V_{\boldsymbol{\theta}}(z')=\frac{\theta_k V_{\boldsymbol{\theta}}(z')} {z-z'}+\mathrm{regular}, 
\]
which in turn imply that
\be\label{Wprim}
 W^{(j)}\lb z\rb V_{\boldsymbol{\theta}}(z')=\sum_{k=0}^\infty \frac{\bigl( W^{(j)}_{-k} V_{\boldsymbol{\theta}}\bigr)(z')} {\lb z-z'\rb^{j-k}},
\ee
where 
\be\label{W0prim}
\bigl( W^{(j)}_0 V_{\boldsymbol{\theta}}\bigr)\lb z\rb=e_j\lb\boldsymbol{\theta}\rb V_{\boldsymbol{\theta}}\lb z\rb,
\ee
and
$e_j(\boldsymbol{\theta})$ is the $j$-th elementary symmetric polynomial in the variables $\boldsymbol{\theta}$. 
The relations (\ref{Wprim}), (\ref{W0prim}) define the primary field $V_{\boldsymbol{\theta}}(z)$ of the $W_N$-algebra. 
Due to the state-field correspondence, to every such primary field we can associate the highest weight vector $|\boldsymbol{\theta}\rangle$ of
a Verma module $\mathsf{M}_{\boldsymbol{\theta}}$ of the $W_N$-algebra.

The Verma module $\mathsf{M}_{\boldsymbol{\theta}}$ is induced from the one-dimensional
module with the basis element $|\boldsymbol{\theta}\rangle$ of the subalgebra of $W_N$-algebra generated by $\{W^{(j)}_{k\ge 0}\}$:
\be\label{hvvright}
 W^{(j)}_0 |\bt\rangle =e_j\lb \boldsymbol{\theta}\rb|\bt\rangle,\qquad 
 W^{(j)}_{k >0}\left|\bt\right\rangle =0.
\ee
The $W_N$-algebra admits a Poincar\'e-Birkhoff-Witt basis \cite{Arakawa}. This means that there is a linear basis in the Verma module $\mathsf{M}_{\boldsymbol{\theta}}$
consisting of the elements $W_{\boldsymbol{\lambda}} |\boldsymbol{\theta}\rangle$, where ${\boldsymbol{\lambda}=\lb \lambda^{(2)},\ldots,\lambda^{(N)}\rb\in\mathbb Y^{N-1}}$ is an $\lb N-1\rb$-tuple of partitions and
\be\label{stPBW}
W_{\boldsymbol{\lambda}} \left|\boldsymbol{\theta}\right\rangle=
W^{(N)}_{-\lambda^{(N)}}\cdots
W^{(2)}_{-\lambda^{(2)}}\left|\boldsymbol{\theta}\right\rangle,
\qquad
W^{(s)}_{-\lambda^{(s)}}:=W^{(s)}_{-\lambda^{(s)}_1}\cdots
W^{(s)}_{-\lambda^{(s)}_{\ell_s}}.
\ee
We are interested in the matrix elements of descendants (3-point functions) of the vertex operator  $V_{\boldsymbol{\theta}}\lb z\rb$:
\be\label{gen3pME}
 \left\langle \boldsymbol{\theta}''\right| W^\dag_{\boldsymbol{\lambda}''} \left(W_{\boldsymbol{\lambda}}  
 V_{\boldsymbol{\theta}}\right)\lb z\rb W_{\boldsymbol{\lambda}'} \left|\boldsymbol{\theta}'\right\rangle,\qquad
 \boldsymbol{\lambda},\boldsymbol{\lambda'},\boldsymbol{\lambda''}
 \in \mathbb Y^{N-1},
\ee
where $\dag$ is an anti-linear involutive anti-automorphism of the $W_N$-algebra uniquely defined by $\bigl (W^{(j)}_{k}\bigr)^{\dag}=W^{(j)}_{-k}$, and
$\langle \bt|$ satisfies the conditions analogous to \eqref{hvvright}:
\be\label{hvvleft}
 \left\langle \bt\right| W^{(j)}_0 =  \left\langle \bt\right| e_j\lb\bt\rb ,\qquad 
 \left\langle \bt\right| W^{(j)}_{k<0}  =0.
\ee

It is well-known that in the case of the Virasoro algebra (i.e. for $N=2$) thanks to the Ward identities all matrix elements \eqref{gen3pME}
can be expressed in terms of the matrix element of  $V_{\boldsymbol{\theta}}\lb z\rb$ between the highest weight vectors,
$\left\langle \boldsymbol{\theta}''\right| V_{\boldsymbol{\theta}}\lb z\rb \left|\boldsymbol{\theta}'\right\rangle$ (3-point function of primaries). In order to simplify these matrix elements as much as possible in the $N\ge 3$ case, we will use the Borcherds identities
\cite{Borch,FBZ,Kac} 
\be\label{WVcom}
[W^{(j)}_p, V\lb z\rb]= \sum_{k=1-j}^{\infty} z^{p-k} {p+j-1 \choose k+j-1} \lb W^{(j)}_{k}V\rb\lb z\rb, 
\ee
\be\label{descVO}
 \left( W^{(j)}_{p} V\right)(z)=\sum_{k=0}^\infty (-z)^k {j+p-1 \choose k} W^{(j)}_{p-k} V\lb z\rb-
\sum_{k=0}^\infty (-z)^{p+j-k-1} {j+p-1 \choose k}  V\lb z\rb W^{(j)}_{1-j+k},
\ee
valid for any descendant $V\lb z\rb$ of the primary vertex operator $V_{\boldsymbol{\theta}}\lb z\rb$ and any $p\in\mathbb Z$ (and, in fact, for any central charge $c$).

The following theorem naturally generalizes the corresponding result for the $W_3$-algebra \cite{BW,KMS} and is valid for any $c$.
\begin{theo}\label{thmgenME}
General matrix elements of the vertex operators of the $W_N$-algebra can be reduced to the following linear combinations:
\be\label{genWME}
 \left\langle \boldsymbol{\theta}''\right| W^\dag_{\boldsymbol{\lambda}''} \left(W_{\boldsymbol{\lambda}}  
 V_{\boldsymbol{\theta}}\right)\lb z\rb W_{\boldsymbol{\lambda}'} \left|\boldsymbol{\theta}'\right\rangle=\sum_{\bs\mu} A_{\bs\mu}\lb z\rb
\left\langle \boldsymbol{\theta}''\right|   V_{\bt}\lb z\rb   W^{(N)}_{-\mu^{(N)}}\cdots
 W^{(3)}_{-\mu^{(3)}} \left|\bt'\right\rangle,
\ee
where the coefficients $A_{\bs\mu}\lb z\rb$ are
labeled by $\boldsymbol{\mu}=(\mu^{(3)},\ldots,\mu^{(N)})\in\mathbb Y^{N-2}$, and the corresponding partitions are restricted so that $\mu^{(j)}_1\le j-2$ for $j=3,\ldots,N$.

\end{theo}

\begin{proof}
First let us move all $W^{(j)}_p$ with $p>0$  in (\ref{gen3pME}) to the right of the vertex operator with the help of  (\ref{WVcom}).
After this procedure we obtain matrix elements of the form (\ref{gen3pME}) but without~$W^\dag_{\boldsymbol{\lambda}''}$.

At the next step, use the identities (\ref{descVO})
to reduce the matrix elements of descendant operators $V\lb z\rb=\left(W_{\tilde{\boldsymbol{\lambda}}}  
 V_{\bt}\right)\lb z\rb$  to those of primary vertex operator $V_{\bt}\lb z\rb$.
Note that for $p<0$ the matrix elements corresponding to the  first sum on the right of (\ref{descVO}) vanish 
due to (\ref{hvvleft}). Therefore after the 2nd step we come to linear combinations of  matrix elements of type
\be\label{1gen3pME}
 \left\langle \boldsymbol{\theta}''\right| 
 V_{\boldsymbol{\theta}}\lb z\rb W_{\boldsymbol{\lambda}} \left|\boldsymbol{\theta}'\right\rangle,
\ee
with so far unrestricted vectors (\ref{stPBW}). 

Finally, let us change the basis (\ref{stPBW}) in the Verma module $\mathsf{M}_{\bt}$.
We will use new generators of the $W_N$-algebra (see \cite{BW} for the $W_3$ case):
\[
 w^{(j)}_p\lb z\rb=\sum_{k=0}^{j}(-z)^k {j \choose k} W^{(j)}_{p-k},\qquad
 \tilde{w}^{(j)}_0\lb z\rb=\sum_{k=1}^{j}(-z)^k {j-1 \choose k-1} W^{(j)}_{-k}.
\]
They satisfy the relations  (for all $p\in\mathbb Z$)
\[
 [V_{\bt}\lb z\rb,  w^{(j)}_p\lb z\rb]=0,\qquad
 [V_{\bt}\lb z\rb,  \tilde{w}^{(j)}_0\lb z\rb]=(W^{(j)}_{0}V_{\bt})\lb z\rb,
\]
which imply the following action formulas:
\be\label{AF1}
 \left\langle \boldsymbol{\theta}''\right| 
 V_{\bt}\lb z\rb  w^{(j)}_p\lb z\rb =0, \qquad p<0\,,
\ee\be \label{AF2}
 \left\langle \bt''\right| 
 V_{\bt}\lb z\rb  w^{(j)}_0\lb z\rb = e_j(\bt'') \left\langle \bt''\right| V_{\bt}\lb z\rb , 
\qquad
 \left\langle \bt''\right|  V_{\bt}\lb z\rb  \tilde{w}^{(j)}_0\lb z\rb =  e_j(\bt) \left\langle \bt''\right| V_{\bt}\lb z\rb.
\ee
The new PBW basis in the Verma module is labeled by $(\boldsymbol{\lambda},\boldsymbol{\mu},\mathbf{k},\tilde{\mathbf{k}})$, 
where 
\begin{align*}
\begin{gathered}
\boldsymbol{\lambda}=\lb\lambda^{(2)},\ldots,\lambda^{(N)}\rb
\in\mathbb Y^{N-1},\quad
\mathbf{k}=\lb k_2,\ldots,k_N\rb\in \mathbb{Z}_{\ge 0}^{N-1},\quad
\tilde{\mathbf{k}}=\lb\tilde{k}_2,\ldots,\tilde{k}_N\rb\in \mathbb{Z}_{\ge 0}^{N-1},\\
\boldsymbol{\mu}=\lb\mu^{(3)},\ldots,\mu^{(N)}\rb\in\mathbb Y^{N-2},\text{ with }\mu^{(j)}_1\le j-2\text{ for }j=3,\ldots,N,
\end{gathered}
\end{align*}
and is given by the vectors
\begin{equation*}\label{newPBW}
 W_{(\boldsymbol{\lambda},\boldsymbol{\mu},\mathbf{k},\tilde{\mathbf{k}})}\left|\bt\right\rangle
 =w^{(N)}_{-\lambda^{(N)}}\cdots
 w^{(2)}_{-\lambda^{(2)}} \bigl(w^{(N)}_{0}\bigr)^{k_N} \bigl(\tilde{w}^{(N)}_{0}\bigr)^{\tilde{k}_N}
 \cdots \bigl(w^{(2)}_{0}\bigr)^{k_2} \bigl(\tilde{w}^{(2)}_{0}\bigr)^{\tilde{k}_2}
 W^{(N)}_{-\mu^{(N)}}\cdots W^{(3)}_{-\mu^{(3)}}
 \left|\bt\right\rangle.
\end{equation*}
Thanks to (\ref{AF1}), (\ref{AF2}), the matrix elements
$\left\langle \boldsymbol{\theta}''\right|   V_{\bt}\lb z\rb W_{(\boldsymbol{\lambda},\boldsymbol{\mu},\mathbf{k},\tilde{\mathbf{k}})}\left|\bt'\right\rangle$ can be expressed in terms of
the matrix elements 
\be\label{redME}
\left\langle \boldsymbol{\theta}''\right|   V_{\bt}\lb z\rb  W_{(\varnothing,\boldsymbol{\mu},\mathbf{0},{\mathbf{0}})}
\left|\bt'\right\rangle=
\left\langle \boldsymbol{\theta}''\right|   V_{\bt}\lb z\rb   W^{(N)}_{-\mu^{(N)}}\cdots
 W^{(3)}_{-\mu^{(3)}} 
 \left|\bt'\right\rangle
\ee
labeled by tuples of partitions $\boldsymbol{\mu}\in\mathbb Y^{N-2}$ which satisfy the above restrictions $\mu^{(j)}_1\le j-2$. 
Note that $\mu^{(j)}$ may be equivalently represented by $j-2$ non-increasing non-negative integers. For example, for $N=4$, the minimal set of matrix elements can be chosen as
\[
\left\langle \boldsymbol{\theta}''\right|   V_{\bt}\bigl( z\bigr)   \bigl( W_{-2}^{(4)}\bigr)^{l_2}\bigl(  W_{-1}^{(4)}\bigr)^{l_1-l_2}\bigl(  W_{-1}^{(3)}\bigr)^{l}
\left|\bt'\right\rangle,\qquad l_1\ge l_2\ge0,\quad l\ge0.
\]
\end{proof}

For generic weights $\bt$, $\bt'$, $\bt''$ the matrix elements (\ref{redME}) can not be related by means of the Ward identities. 
However, if one of these weights is of semi-degenerate type (to be discussed in the next subsection) then all these matrix elements can be expressed in terms of $\left\langle \boldsymbol{\theta}''\right|   V_{\bt}\lb z\rb \left|\bt'\right\rangle$, just as in the case of the Virasoro algebra.

\subsection{Semi-degenerate representations}

We will need special reducible Verma modules with $\bt=a \bs h_1$, where $a$ is a complex number and $\bs h_s$, $s=1,\ldots,N$, are the weights
of the first fundamental representation of $\mathfrak{sl}_N$ with the components 
\be\label{hweight}
 h_s^{(k)}=\delta_{sk}-1/N,\qquad k=1,\ldots,N.
\ee
The irreducible quotient with the highest weight $\bt=a \bs h_1$ is called semi-degenerate representation.
We have $N-2$ relations 
\[
\left[W^{(r)}_{-1}-{N-2\choose r-2} \lb -\frac aN\rb^{r-2} W^{(2)}_{-1}\right]\left|a \bs h_1\right\ra=0, \qquad r=3,\ldots, N,
\]
corresponding to singular vectors on the first level of the $L_0$-gradation in the Verma module.
All the relations needed for derivation on the $p$-th level are given by
\be\label{semdegrel}
 \left[W^{(r)}_{-p}+(-1)^{r+p} \sum_{s=2}^{p+1}  {N-s\choose r-s} {r-s-1\choose p-s+1} \left(\frac{a}{N}\right)^{r-s} W^{(s)}_{-p}\right]\left|a \bs h_1\right\ra=0, 
 \quad 2\le p+1<r\le N,
\ee
and correspond to factoring out  different proper submodules in the Verma module. The derivation of these relations is given in Appendix~A.

The following proposition shows how the relations (\ref{semdegrel}) can be used for further reduction of the matrix elements appearing in (\ref{genWME}).

\begin{prop}\label{propsdvo}
Matrix elements of the semi-degenerate vertex operator $V_{a \bs h_1}\lb z\rb$ and its descendants can be expressed in terms of
the primary matrix element
\be\label{normVO}
\left\la \bt'\right| V_{a \bs h_1}\lb z\rb\left|\bt\right\ra=\mathcal{N}\lb\bt',a \bs h_1,\bt\rb z^{\Delta_{\bt'}-\Delta_{a \bs h_1}-\Delta_{\bt}},
\ee
where $\Delta_{\bt}=-e_2\lb \bt\rb=\bt^2/2$.
\end{prop}
\begin{proof}
Theorem~\ref{thmgenME} allows us to start the reduction procedure from the matrix elements of the form
\be\label{ME_WjpW}
\left\la \bt'\right| V_{a \bs h_1}(z) W^{(j)}_{-p} \mathcal{W} \left|\bt\right\ra, \qquad 1\le p\le j-2,
\ee
where $\mathcal{W}$ is a product of the generators of the $W_N$-algebra.
We will reduce such matrix elements to 
$\left\la \bt'\right| V_{a \bs h_1}\lb z\rb  \widetilde {\mathcal{W}}\left|\bt\right\ra$ with $\widetilde {\mathcal{W}}$
having $\deg_W \widetilde {\mathcal{W}}< j-1+ \deg_W \mathcal{W}$, cf (\ref{degW}), (\ref{crW}).
Since $\deg_W  W^{(j)}_{-p}  =j-1$,  it then suffices to use induction in $\deg_W \mathcal{W}$. 

The identity \eqref{WVcom} can be rewritten for  $V\lb z\rb=V_{\boldsymbol{\theta}}\lb z\rb$ and any $p\in\mathbb Z$ as
\be\label{comWprim}
\left[W^{(j)}_{-p}, V_{\boldsymbol{\theta}}\lb z\rb\right]=z^{-p} \sum_{k=0}^{j-1} z^k {j-p-1 \choose j-k-1} \lb W^{(j)}_{-k}V_{\boldsymbol{\theta}}\rb\lb z\rb.
\ee
This commutation relation allows to transform (\ref{ME_WjpW}) into a linear combination of matrix elements
\be\label{ME_WjpVW}
\left\la \bt'\right| \bigl(W^{(j)}_{-p} V_{a\bs h_1}\bigr) \lb z\rb  \mathcal{W} \left|\bt\right\ra
\ee
with $1\le p\le j-1$.
Moreover, we can exclude $\left\la \bt'\right| \bigl(W^{(j)}_{-(j-1)} V_{a \bs h_1}\bigr) \lb z\rb  \mathcal{W} \left|\bt\right\ra$ from this set of matrix elements
using the relation (\ref{comWprim}) with $p=0$, which produces one more linear combination of matrix elements (\ref{ME_WjpVW}). It can be found since  
$\left\la \bt'\right| \bigl[W^{(j)}_{0}, V_{\boldsymbol{\theta}}\lb z\rb\bigr] \mathcal{W} \left|\bt\right\ra$ may be computed independently using 
(\ref{hvvright}), (\ref{hvvleft}) and the fact that 
\[
W^{(j)}_0 \mathcal{W}=\mathcal{W} W^{(j)}_0 +\widetilde {\mathcal{W}},\qquad \deg_W \widetilde {\mathcal{W}}< \deg_W (W^{(j)}_0 \mathcal{W}),
\]
by the induction assumption.  Thus the problem is reduced to finding matrix elements (\ref{ME_WjpVW}) for
$1\le p\le j-2$.

The next step is to use (\ref{semdegrel}). The identity (\ref{descVO}) then produces matrix elements of the form
$\left\la \bt'\right| V_{a \bs h_1}\lb z\rb  \widetilde {\mathcal{W}}\left|\bt\right\ra$ for $\widetilde {\mathcal{W}}$
with $\deg_W \widetilde {\mathcal{W}}< j-1+ \deg_W \mathcal{W}$.  These elements are known by the induction assumption.

The starting matrix element of the induction procedure is $\left\la \bt'\right| V_{a \bs h_1}\lb z\rb\left|\bt\right\ra$. It can be calculated, up to a normalization factor, from the
relation $\left\la \bt'\right| \bigl(L_{-1} V_{a \bs h_1}\bigr)\lb z\rb\left|\bt\right\ra=\partial_z \left\la \bt'\right| V_{a \bs h_1}\lb z\rb\left|\bt\right\ra$, with $L_p=-W^{(2)}_p$. This yields (\ref{normVO}).
\end{proof}

\subsection{Degenerate representations}

We will need even more special irreducible representation with $a=1$, i.e. $\bt=\bs h_1$. It corresponds to 
the first fundamental representation of $\mathfrak{sl}_N$. Another important representation 
has $\bt=-\bs h_N$ and corresponds to the last fundamental representation of $\mathfrak{sl}_N$. 
Irreducible representations with highest weights $\bt=\bs h_1$  and $\bt=-\bs h_N$  are called (completely) degenerate representations.

One may expect that all the normalization coefficients $\mathcal{N}\lb \bt',a\bs h_1,\bt\rb$ (structure constants) are non-zero for generic $a$, $\bt$, $\bt'$. However, 
in the case of degenerate vertex operator with $a=1$ there are additional restrictions on the possible values of $\bt'$ to have non-vanishing
$\mathcal{N}\lb \bt',a\bs h_1,\bt\rb$ due to 
additional singular vectors in the Verma module $\mathsf{M}_{\bs h_1}$. 
\begin{prop}
The fusion rule for $V_{\bs h_1}\lb z\rb$ is
\be\label{fusplus}
V_{\bs h_1}(z)\left|\bt\right\ra=\sum_{s=1}^N \mathcal{N}\lb \bt+\bs h_s,\bs h_1,\bt\rb z^{\Delta_{\bt+\bs h_s}-\Delta_{\bs h_1}-\Delta_{\bt}} \bigl[\left|\bt+\bs h_s\right\ra+O\lb z\rb\bigr].
\ee
Similarly, the fusion rule for $V_{-\bs h_N}\lb z\rb$ is
\be\label{fusminus}
V_{-\bs h_N}\lb z\rb\left|\bt\right\ra=\sum_{s=1}^N \mathcal{N}\lb \bt-\bs h_s,\bs h_1,\bt\rb z^{\Delta_{\bt-\bs h_s}-\Delta_{\bs h_1}-\Delta_{\bt}} \bigl[\left|\bt-\bs h_s\right\ra+O\lb z\rb\bigr].
\ee
\end{prop}
\begin{proof}
The derivation of these fusion rules is given in the Appendix~B.
\end{proof}

In what follows, we will use the projector $\mathcal{P}_\bt$ to the irreducible module with the highest weight $\bt$
and a special notation for the degenerate vertex operators restricted to a particular fusion channel:
\be\label{psifr}
 \psi_{s,\bt}\lb y\rb=\mathcal{P}_{\bt+\bs h_s} V_{\bs h_1}\lb y\rb\mathcal{P}_{\bt},\qquad
 \bar \psi_{s,\bt}\lb y\rb=\mathcal{P}_{\bt-\bs h_s} V_{-\bs h_N}\lb y\rb \mathcal{P}_{\bt}.
\ee
Sometimes we will use shorthand notations $\psi_{s}(y)$ and $\bar \psi_{s}(y)$ if a particular $\bt$ is understood.
Explicitly, the fusion rules for them are
\be\label{fuspsiplus}
\psi_{s}\lb y\rb\left|\bt\right\ra= \mathcal{N}\lb\bt+\bs h_s,\bs h_1,\bt\rb y^{\Delta_{\bt+\bs h_s}-\Delta_{\bs h_1}-\Delta_{\bt}} \bigl[\left|\bt+\bs h_s\right\ra+O\lb y\rb\bigr],
\ee
\be\label{fuspsiminus}
\bar \psi_{s}\lb y\rb\left|\bt\right\ra= \mathcal{N}\lb \bt-\bs h_s,\bs h_1,\bt\rb y^{\Delta_{\bt-\bs h_s}-\Delta_{\bs h_1}-\Delta_{\bt}} \bigl[\left|\bt-\bs h_s\right\ra+O\lb y\rb\bigr].
\ee
The singular parts of their OPEs become
\be\label{OPEpsipsibar}
\psi_s\lb z\rb \bar \psi_{s'}\lb w\rb\sim  \frac{\delta_{s,s'}}{\lb z-w\rb^{(N-1)/N}}, \qquad
\psi_s\lb z\rb \psi_{s'}\lb w\rb\sim 0,\qquad \bar \psi_s\lb z\rb \bar\psi_{s'}\lb w\rb\sim 0,
\ee
provided we choose normalizations so that $\mathcal{N}\lb \bt,\bs h_1,\bt+\bs h_s\rb=\mathcal{N}^{-1}\lb \bt+\bs h_s,\bs h_1,\bt\rb$.

\section{Conformal blocks of $W$-algebras and their  properties\label{sec4}}
\subsection{Hypergeometric conformal blocks}

Below we will need special conformal blocks which can be expressed in terms of hypergeometric functions, and their properties.
In this subsection, the non-vanishing 3-point functions of semi-degenerate vertex operators are normalized by $\mathcal{N}\lb \bt',a\bs h_1,\bt\rb=1$, but later this normalization will be changed to a more convenient one.
In what follows, we will use a shorthand notation $V_a\lb z\rb:=V_{a \bs h_1}\lb z\rb$ for the semi-degenerate vertex operators.

\begin{theo}
The following conformal blocks have hypergeometric expressions:
\be\label{HG0}
 \left\la \bt_{\infty}\right|V_a \lb z\rb \psi_s\lb y\rb\left| \bt_{0}\right\ra=
 z^{\Delta_4-\Delta_3-\Delta_2-\Delta_1}
 \lb  y/z\rb^{\Delta_{1,s}-\Delta_2-\Delta_1} \lb 1- y/z\rb^{-a/N} \mathcal{G}_s\lb y/z\rb,
\ee\be\label{HGi}
 \left\la\bt_{\infty}\right|\psi_s\lb y\rb V_a \lb z\rb\left|\bt_{0}\right\ra=
 y^{\Delta_4-\Delta_3-\Delta_2-\Delta_1}
 \lb z/y\rb^{\Delta_{4,s}-\Delta_1-\Delta_3}  \lb 1-z/y\rb^{-a/N} \mathcal{G}'_s\lb z/y\rb,
\ee
where, recalling the notation $\Delta_{\bt}=\bt^2/2$ for conformal weights, we have
\[
 \Delta_1=\Delta_{\bt_{0}},\qquad 
 \Delta_2=\Delta_{\bs h_1}=\tfrac{N-1}{2N},\qquad
 \Delta_3=\Delta_{a \bs h_1}=a^2\cdot\tfrac{N-1}{2N},\qquad
 \Delta_4=\Delta_{ \bt_{\infty}},
 \]\[
 \Delta_{1,s}=\Delta_{ \bt_{0}+\bs h_s},\qquad 
 \Delta_{4,s}=\Delta_{\bt_{\infty}-\bs h_s},
\] 
\[
\mathcal{G}_s\lb x\rb={}_NF_{N-1}\left( \genfrac{}{}{0pt}{0}
{\{(N-a-1)/N+\theta_{0}^{(s)}-\theta_{\infty}^{(k)}\}_{k=\overline{1,N}}}
{\{1+\theta_{0}^{(s)}-\theta_{0}^{(k)}\}_{k=\overline{1,N};\, k\ne s}}\,\Biggl| x\right),
\]\[
\mathcal{G}'_s\lb x\rb={}_NF_{N-1}
\left( \genfrac{}{}{0pt}{0}{
\{(N-a-1)/N+\theta_{0}^{(k)}-\theta_{\infty}^{(s)}\}_{k=\overline{1,N}}}
{\{1+\theta_{\infty}^{(k)}-\theta_{\infty}^{(s)}\}_{k=\overline{1,N};\, k\ne s}}\,\Biggl | x \right).
\]
\end{theo}

\begin{proof}
The idea of the proof is to use rigidity of the 3-point Fuchsian system (\ref{FS3p}). 
We claim that the solution of the system 
can be given in terms of 4-point conformal blocks with degenerate fields $\psi_j\lb y\rb$ and a semidegenerate field $V_{a}\lb 1\rb=V_{a \bs h_1}\lb 1\rb$:
\[
\tilde \Phi_{jm}\lb y\rb=K_{jm}\left  \langle -\boldsymbol{\theta}_\infty +\bs h_m \right| \psi_j\lb y\rb V_a\lb 1\rb \left|\boldsymbol{\theta}_0 \right\rangle
\underset{y\to\infty}{=}K_{jm}  y^{-\theta_\infty^{(j)}+\delta_{jm}-1}\bigl[1+ O\lb y^{-1}\rb \bigr]\,,
\]
where we used fusion rules (\ref{fuspsiplus}) to derive the asymptotics as $y\to \infty$. The constants  $K_{jm}$ will be fixed later.
The leading asymptotics of the matrix elements of $\tilde{\Phi}\lb y\rb$ as $y\to \infty$ can also be rewritten as
\[
\tilde \Phi_{jm}\lb y\rb =  y^{-\theta_\infty^{(j)}}\bigl[\delta_{jm} K_{jj}+ O\lb y^{-1}\rb \bigr]\,.  
\]

Each column $m$ of the matrix $\tilde \Phi\lb y\rb$ constitutes a basis in the $N$-dimensional space $\mathcal{V}_m$ of conformal blocks with 
fixed external weights labeled by $\bt_0$, $\bt_1=a \bs h_1$, $\bs h_1$, $-\bt_\infty+\bs h_m$. This space is invariant under analytic continuation in $y$ around 
$z_0=0$, $z_1=1$ and $z_2=\infty$
giving the monodromies $M^{(m)}_0$, $M^{(m)}_1$ and  $M^{(m)}_\infty$, respectively.
The spectra of these monodromy matrices are independent of $m$ and can be found from the conformal dimensions of fields using (\ref{fuspsiplus}).
Namely, all these bases are associated to different channels corresponding to the fusion of the degenerate field $\psi_j(y)$ with 
a generic primary at $z_2=\infty$. 

There are two more bases in each of $\mathcal{V}_m$ associated to the channels 
corresponding to the fusion of the degenerate field $\psi_j(y)$ with the fields at $z_0=0$ and $z_1=1$, respectively.
The leading terms (up to constant prefactors) of the basis conformal blocks for each $\mathcal{V}_m$ at $z_0=0$ are 
$y^{\boldsymbol\theta_0}=(y^{\theta_0^{(1)}},\ldots,y^{\theta_0^{(N)}})$.
Similarly, the leading terms of the basis conformal blocks for each $\mathcal{V}_m$ at $z_1=1$ are 
$\lb y-1\rb^{\boldsymbol\theta_1}=
\bigl( \lb y-1\rb^{\theta_1^{(1)}},\ldots,
\lb y-1\rb^{\theta_1^{(N)}}\bigr)$\footnote{Since the spectrum
of $\Theta_1$ is degenerate of spectral type $\lb N-1,1\rb$, actually there is an ambiguity in the choice of the basis in the space of conformal blocks:
taking their differences it is possible to choose another basis in which the leading terms  are multiplied 
by $\lb y-1\rb^k$, $k\in \mathbb{Z}_{>0}$.}.
These two bases are distinguished by the property of having diagonal monodromies $\exp 2\pi i \Theta_0$ and $\exp 2\pi i \Theta_1$ under analytic continuation around $z_0$ and $z_1$, respectively.
In the initial basis, which is distinguished by the property of having diagonal monodromy around $z_2=\infty$, 
we have
\[
M^{(m)}_0= C^{(m)}_0 e^{2\pi i \Theta_0} {\bigl( C^{(m)}_0\bigr)}^{-1},\quad 
M^{(m)}_\infty= e^{2\pi i \Theta_\infty},\quad
M^{(m)}_1= C^{(m)}_1 e^{2\pi i \Theta_1}\bigl(C^{(m)}_1\bigr)^{-1},
\]
where the diagonal matrices $\Theta_0$, $\Theta_1$ and $\Theta_\infty$ are given by (\ref{FS3pTh}).

{}It follows from Lemma~\ref{KatzLemma} that the triples $\bigl( M^{(m)}_0, M^{(m)}_1, M^{(m)}_\infty\bigr)$ of monodromy matrices  for different $m$ are related by an overall similarity transformation. Since $ M^{(m)}_\infty$ are already diagonal for all $m$ and have simple spectrum,
the remaining freedom is given by the conjugation by diagonal matrices.   
Such type of similarity transformations corresponds to choosing the coefficients $K_{jm}$, $j=1,\ldots,N$, for each $m$. Fix these coefficients
so that all triples $\bigl(M^{(m)}_0, M^{(m)}_1, M^{(m)}_\infty\bigr)$ coincide 
with the triple $\bigl(M_0, M_1, M_\infty\bigr)$ of monodromy matrices  of the actual fundamental 
solution $\Phi(y)$ given by 
(\ref{FS3psol}).
Then the elements of the matrix $\Phi\lb y\rb^{-1} \tilde\Phi\lb y\rb$ are given by single-valued meromorphic functions with the only possible poles at $0$, $1$ or $\infty$. 
However, $\Phi\lb y\rb$ and $\tilde\Phi\lb y\rb$ have the same local monodromy exponents, hence this matrix is in fact constant. From the normalization of $\Phi\lb y\rb$ it follows that the constant matrix is diagonal and may be chosen as the identity matrix, in which case $K_{jj}=1$ and
$\Phi\lb y\rb=\tilde\Phi\lb y\rb$. 

The diagonal elements of the latter relation give the hypergeometric representation \eqref{HGi} with $z=1$;
arbitrary $z$ may be obtained by a scale transformation of 
conformal blocks. The proof of \eqref{HG0} is completely analogous.
\end{proof}

Conformal blocks (\ref{HG0}) are defined as convergent series for $|y|<|z|$. Their analytic continuation in $y$ to the region with $|y|>|z|$ can be compared to conformal blocks (\ref{HGi}), given by convergent series in the latter domain.
The relation between the two sets of conformal blocks may be expressed with the help of well-known $_NF_{N-1}$ connection formulas, see e.g. \cite{Norlund}:
\be\label{FtrME}
  \left\la \bt_{\infty}\right|V_a \lb z\rb \psi_l\lb y.\xi_z\rb\left| \bt_{0}\right\ra=
  \sum_{j=1}^N e^{-i\pi((N-1)/N+\theta_{0}^{(l)}-\theta_{\infty}^{(j)})} \tilde F^{[\infty,0]}_{lj}\lb \bt_{\infty},a,\bt_{0} \rb
\left\la\bt_{\infty} \right|\psi_j\lb y\rb V_a \lb z\rb\left| \bt_{0}\right\ra.
\ee
Here $y.\xi_z$ stands for the analytic continuation along a contour $\xi_z$ going around $z$ in the positive  direction and $\tilde F^{[\infty, 0]}$ is the \textit{fusion matrix} ``$0\to \infty$'' with the elements
\[
 \tilde F^{[\infty,0]}_{lj}\lb \bt_{\infty},a,\bt_{0}\rb=
 \frac{ \prod_{k(\ne l)} \Gamma(1+\theta_{0}^{(l)}-\theta_{0}^{(k)}) \cdot \prod_{k(\ne j)}\Gamma(\theta_{\infty}^{(j)}-\theta_{\infty}^{(k)})}
  {\prod_{k(\ne l)} \Gamma((1+a)/N+\theta_{\infty}^{(j)}-\theta_{0}^{(k)}) \cdot
  \prod_{k(\ne j)} \Gamma((N-1-a)/N+\theta_{0}^{(l)}-\theta_{\infty}^{(k)}) },
\]
where $\Gamma\lb x\rb$ is the gamma function.

In fact, the fusion matrix in  (\ref{FtrME}) will not change if instead of the vectors 
$\left| \bt_{0}\right\ra$ and $\left| \bt_{\infty}\right\ra$ we use any of their descendants, obtained by the action of the creation operators $W^{(j)}_k$, ${k<0}$. This is due to the fact that all such conformal blocks can be obtained  by an action of differential operators in $y$ and $z$, which
commutes (intertwines) with the crossing symmetry transformation (\ref{FtrME}). Combining the vertex operators
$\psi_l\lb y\rb$, $l=1,\ldots,n$, into a column matrix $\Psi\lb y\rb$, one may therefore rewrite (\ref{FtrME}) as an operator relation
\be\label{Vpsit}
 \mathcal{P}_{\bt_{\infty}}V_{a}\lb z\rb\Psi\lb y.\xi_z\rb \mathcal{P}_{\bt_{0}}
 =     B^{-1}\lb \bt_{0}\rb\tilde F^{[\infty,0]}\lb\bt_{\infty},a,\bt_{0}\rb  B'\lb \bt_{\infty}\rb \cdot
    \mathcal{P}_{\bt_{\infty}} \Psi\lb y\rb V_{a}\lb z\rb \mathcal{P}_{\bt_{0}},
\ee
where the braiding matrices $B$ and $B'$ are diagonal and their non-zero elements are given by
\be\label{MEBBp}
 B_{ll}\lb \bt_{0}\rb=\exp{i \pi \theta_{0}^{(l)}},\qquad
 B'_{jj}\lb \bt_{\infty}\rb=\exp{i \pi \bigl(\theta_{\infty}^{(j)}-\tfrac{N-1}{N}\bigr)}\,.
\ee

Analytic continuation of $\psi_l\lb y\rb\left|\bt_0\right\ra$ in $y$ around $0$ in the positive direction leads to multiplication by the diagonal braiding matrix
$B^2(\bt_0)$. Indeed,
\[
B^2_{ll}\lb \bt_{0}\rb=
 \exp\left\{{2\pi i}\bigl(\Delta_{\bt_{0}+ \bs h_l}-\Delta_{\bt_{0}}-\Delta_{ \bs h_l}\bigr)\right\}=\exp{2 \pi i \theta_{0}^{(l)}}.
\]
Taking into account that the braiding matrix is the same for  descendants of $\left|\bt_{0}\right\ra$, we write the braiding relation as
\be\label{psiB}
 \Psi\lb ye^{2\pi i}\rb\mathcal{P}_{\bt_{0}}=B^2\lb\bt_{0}\rb\cdot  \Psi\lb y\rb\mathcal{P}_{\bt_{0}}.
\ee
Similarly, the analytic continuation of $\left\la \bt_{\infty}\right|\psi_j\lb y\rb$ in $y$ around $\infty$ in the 
\textit{negative} direction leads to  multiplication by the diagonal braiding matrix
${B'}^2\lb \bt_{\infty}\rb$, whose non-vanishing elements are
\[
{B'}^{2}_{jj}\lb\bt_{\infty}\rb=
 \exp\left\{ {2\pi i}\bigl(\Delta_{\bt_{\infty}}-\Delta_{\bt_{\infty}- \bs h_j}-\Delta_{ \bs h_j}\bigr)\right\}=
 \exp{2 \pi i  \bigl(\theta_{\infty}^{(j)}-\tfrac{N-1}{N}\bigr)}.
\]
This leads to
\be\label{psiBp}
\mathcal{P}_{\bt_{\infty}}\Psi\lb ye^{2 \pi i}\rb= {B'}^2\lb\bt_{\infty}\rb\cdot\mathcal{P}_{\bt_{\infty}}\Psi\lb y\rb.
\ee

Finally, let us introduce the following formal transformations which are a consequence of the definition (\ref{psifr}):
\[
 \psi_s\lb y\rb\mathcal{P}_{\boldsymbol{\sigma}}=\mathcal{P}_{\boldsymbol{\sigma}+\bs h_s}\psi_s\lb y\rb=
 \nabla_{\boldsymbol{\sigma},s} \mathcal{P}_{\boldsymbol{\sigma}}\psi_s\lb y\rb,
\]
where $\nabla_{\boldsymbol{\sigma},s}$ is the shift operator defined by $\nabla_{\boldsymbol{\sigma},s}\, \mathcal{F}\lb \boldsymbol{\sigma}\rb=\mathcal{F}(\boldsymbol{\sigma}+\bs h_s)$
for any function $\mathcal{F}$ depending on $\boldsymbol{\sigma}$.
We combine the shifts $\nabla_{\boldsymbol{\sigma},s}$, $s=1,\ldots,N$, into the diagonal matrix
$
 \nabla_{\boldsymbol{\sigma}}=\mathrm{diag}\lb \nabla_{\boldsymbol{\sigma},1},\ldots,\nabla_{\boldsymbol{\sigma},N}\rb
$,
to write compactly
\be\label{PsiP}
  \Psi\lb y\rb\mathcal{P}_{\boldsymbol{\sigma}}= \nabla_{\boldsymbol{\sigma}} \mathcal{P}_{\boldsymbol{\sigma}}\Psi\lb y\rb.
\ee
Note that there is a useful relation between the two types of braiding,
\be\label{BBp}
 B\lb \bs \sigma\rb \nabla_{\bs \sigma}= \nabla_{\bs \sigma}  {B'}\lb \bs \sigma\rb.
\ee

\subsection{Normalization of vertex operators and properties 
of fusion matrices}

Let us change the normalization of the vertex operators in (\ref{normVO}): instead of $\cN \lb \bs\sigma',a\bs h_1,\bs\sigma\rb=1$, we will use 
\be\label{Nnormcoef}
\cN\lb \bs\sigma',a\bs h_1,\bs\sigma\rb=
\frac{\prod_{l,j} G(1-a/N+\sigma^{(l)}-\sigma'^{(j)})} 
 {\prod_{k<m}G(1+\sigma^{(k)}-\sigma^{(m)})\,
 G(1-\sigma'^{(k)}+\sigma'^{(m)})},
\ee
or
\be\label{Nchecknormcoef}
\check{\cN}\lb \bs\sigma',a\bs h_1,\bs\sigma\rb=
 \frac{\prod_{l,j} G(1+a/N-\sigma^{(l)}+\sigma'^{(j)})} 
 {\prod_{k<m}G(1+\sigma^{(k)}-\sigma^{(m)})\,
 G(1-\sigma'^{(k)}+\sigma'^{(m)})},
\ee
where $G\lb x\rb$ is the Barnes $G$-function. 
Both expressions can be employed for the generic semi-degenerate vertex operators (we will use the first one),
but for the degenerate field $\psi_s\lb z\rb$ we can use only the second expression since the first one 
becomes singular:
\be\label{normpsi}
\left\la \bs\sigma + \bs h_s \right| \psi_s\lb 1\rb \left| \bs\sigma \right\ra = \check{\cN}\lb \bs\sigma+\bs h_s,\bs h_1,\bs\sigma\rb.
\ee
In order to preserve the operator product expansion (\ref{OPEpsipsibar}) between $\psi_s\lb z\rb$ and $\bar\psi_s\lb w\rb$, we choose the following normalization for the degenerate field 
$\bar\psi_s\lb w\rb$:
\be\label{normpsibar}
\left\la \bs\sigma \right| \bar\psi_s\lb 1\rb \left| \bs\sigma + \bs h_s \right\ra = \check{\cN}^{-1}\lb \bs\sigma+\bs h_s,\bs h_1,\bs\sigma\rb.
\ee

Straightforward calculation yields
\[\begin{gathered}
\mathsf{N}_{lj}\lb \bs\sigma',a,\bs\sigma\rb:= \frac{\cN\lb \bs\sigma',a \bs h_1,\bs\sigma+\bs h_l\rb{\check \cN}\lb \bs\sigma+\bs h_l, \bs h_1 ,\bs\sigma\rb}
      {{\check \cN}\lb \bs\sigma',  \bs h_1 ,\bs\sigma'-\bs h_j\rb \cN\lb \bs\sigma'-\bs h_j, a  \bs h_1 ,\bs\sigma\rb}=
\\
     = \frac{\prod_{k(\ne j)}\Gamma((N-1-a)/N+\sigma^{(l)}-\sigma'^{(k)}))}
           {\prod_{k(\ne l)}\Gamma((N-1-a)/N+\sigma^{(k)}-\sigma'^{(j)}))}\cdot
           \frac{\prod_{k(\ne l)}\Gamma(\sigma^{(k)}-\sigma^{(l)})}
           {\prod_{k(\ne j)}\Gamma(\sigma'^{(j)}-\sigma'^{(k)})}.
\end{gathered}           
\]
Therefore the new fusion matrix $F\lb \bs\sigma',a,\bs\sigma \rb\equiv F^{[\infty,0]}\lb \bs\sigma ',a,\bs\sigma\rb$ is given by  
\be \label{Fmat}
F_{lj}\lb \bs\sigma',a,\bs\sigma\rb=\mathsf{N}_{lj}\lb \bs\sigma',a,\bs\sigma\rb \tilde  F^{[\infty,0]}_{lj}\lb\bs\sigma',a,\bs\sigma\rb
=\prod_{k(\ne l)} \frac{ \sin\pi((a+1)/N+\sigma'^{(j)}-\sigma^{(k)})}
{\sin\pi(\sigma^{(k)}-\sigma^{(l)})}.
\ee
Its trigonometric form is the crucial advantage of the modified normalization of vertex operators. Note that the inverse fusion matrix 
admits a simple expression:
\be\label{FiF} 
F^{-1}\lb \bs\sigma', a, \bs\sigma\rb= F\lb-\bs\sigma, a, -\bs\sigma'\rb.
\ee  The braiding matrices $B\lb \bs\sigma\rb$, $B'\lb\bs\sigma\rb$ defined by (\ref{MEBBp}) are not affected by the change of normalization. Therefore, the fusion relation (\ref{Vpsit}) transforms into 
\be\label{Vpsi}
 \mathcal{P}_{\bs\sigma'}V_{a}\lb z\rb\Psi\lb y.\xi_z\rb  \mathcal{P}_{\bs\sigma}
 =     B^{-1}\lb \bs\sigma\rb F\lb \bs\sigma',a,\bs\sigma\rb  B'\lb \bs\sigma'\rb \cdot
    \mathcal{P}_{\bs\sigma'} \Psi\lb y\rb V_{a}\lb z\rb \mathcal{P}_{\bs\sigma}.
\ee

In what follows, we will also need additional easily verifiable properties of $F\lb \bs \sigma',a,\bs \sigma\rb$.

\begin{prop} The fusion matrix $F\lb \bs \sigma',a,\bs \sigma\rb$ satisfies the following shift transformations:
\begin{equation}\label{shiftid}
\begin{gathered}
 F\lb \bs\sigma',a,\bs\sigma\pm \bs h_m\rb = D_m^{N-1} F\lb \bs\sigma',a\pm 1,\bs\sigma\rb, \\
 F\lb \bs\sigma'\pm \bs h_m,a,\bs\sigma\rb = F\lb\bs\sigma',a\mp 1,\bs\sigma\rb D_m^{N-1},
\end{gathered}
\end{equation}
where $D_m$ is a diagonal matrix with matrix elements $\lb D_m\rb_{kj}=\lb -1\rb^{\delta_{km}}\delta_{kj}$
and the weight vectors $\bs h_m$ are given by (\ref{hweight}).
Compositions of the transformations (\ref{shiftid}) give the following identities:
\be\label{shiftroot}
\begin{gathered}
 F\lb \bs \sigma',a,\bs\sigma+\bs h_m-\bs h_s\rb = D^{N-1}_m D^{N-1}_s F\lb \bs \sigma',a,\bs\sigma\rb,\\
 F\lb \bs\sigma'+\bs h_m-\bs h_s,a,\bs\sigma\rb = F\lb \bs\sigma',a,\bs\sigma\rb D^{N-1}_m D^{N-1}_s .
\end{gathered}
\ee
\end{prop}

The main result of the paper is an explicit solution of the  Riemann-Hilbert problem for semi-degenerate Fuchsian systems in terms of $W_N$-conformal blocks.
We will find the monodromy matrices of the proposed solution and compare them with the monodromy matrices (\ref{MkWk}) 
parameterized by Proposition~\ref{propparam} with the help of Lemma~\ref{KatzLemma}. For such comparison, the following proposition will be useful.
\begin{prop}\label{XYME}
The matrix $F\lb \bs \sigma',a,\bs \sigma\rb$ given by (\ref{Fmat}) and the matrices $W_k$ given by (\ref{Wm}) are related by
\be\label{FW}
 F\lb \bs \sigma_k,a_k-1,\bs \sigma_{k-1}\rb= X_k W_k Y_k^{-1},
\ee
where the diagonal matrices $X_k$ and $Y_k$ are defined by their diagonal elements $x_k^{(s)}$ and $y_k^{(s)}$, $s=1,\ldots,N$, respectively:
\begin{align}\label{xME}
 \bigl(x_k^{(s)}\bigr)^{-1}=&\,e^{i\pi(N-1) \sigma^{(s)}_{k-1}}
 \prod_{m=1}^N \sin \pi(\hat \sigma_k^{(m)}-\sigma^{(s)}_{k-1})
 \prod_{m(\ne s)} \sin \pi(\sigma^{(m)}_{k-1}-\sigma^{(s)}_{k-1}),
\\
\label{yME}
 \bigl(y_k^{(s)}\bigr)^{-1}=&\,\;e^{{i}\pi (N-1)\hat\sigma_k^{(s)}}\, \prod_{m=1}^N\sin \pi(\hat\sigma_k^{(s)}-\sigma^{(m)}_{k-1})
 \prod_{m(\ne s)}\sin \pi(\hat\sigma_k^{(m)}-\hat\sigma_k^{(s)}),
\end{align}
and we use the notation $\hat \sigma_k^{(s)}=\sigma_k^{(s)}+a_k/N$.
\end{prop}

\subsection{Monodromy of conformal blocks}
This subsection is devoted to the computation of monodromy of multipoint conformal blocks using the 
transformation properties (\ref{Vpsi}), (\ref{psiB})--(\ref{BBp}).

Consider the following column $\mathcal{F}_m\lb \bs\sigma|\,y\rb$, $\bs\sigma=\lb\bs\sigma_1,\ldots,\bs\sigma_{n-3}\rb$ of  conformal blocks: 
\[
\mathcal{F}_m\lb\bs \sigma |\, y\rb=C_m
 \left\langle -\boldsymbol{\theta}_{n-1} +\bs h_m\right| \Psi\lb y\rb V_{a_{n-2}}\lb z_{n-2}\rb 
 \mathcal{P}_{\boldsymbol{\sigma}_{n-3}}
 V_{a_{n-3}}\lb z_{n-3}\rb\mathcal{P}_{\boldsymbol{\sigma}_{n-4}} \cdots \mathcal{P}_{\boldsymbol{\sigma}_1}
 V_{a_{1}}\lb z_{1}\rb \left|\boldsymbol{\theta}_0\right\rangle,
\]
where $C_m$ is the diagonal matrix with matrix elements $\lb C_m\rb_{jj}=\lb-1\rb^{N\delta_{jm}}$.
The advantage of incorporating  $C_m$ into the definition of $\mathcal{F}_m\lb\bs\sigma|\,y\rb$ is that its monodromy becomes independent of $m$.
Denote by $\mathcal{F}_m\lb \bs\sigma|\,y.\xi_{[k]}\rb$ the analytic continuation of $\mathcal{F}_m\lb \bs\sigma|\,y\rb$ in $y$ along the path $\xi_{[k]}:=\xi_0\cdots\xi_k$ encircling $z_{0},\ldots,z_k$.
We have \[
 \mathcal{F}_m\lb \bs\sigma|\,y.\xi_{[k]}\rb =  \hat M_{[k]} \ \mathcal{F}_m\lb \bs\sigma|\,y\rb,
\]
where the operator-valued monodromy matrix $\hat M_{[k]}$ is
\be\label{MVBV}
\hat M_{[k]} = \hat V'_{[k]} B^{2}(\boldsymbol{\sigma}_{k}) {(\hat V'_{[k]})}^{-1},
\ee
\[\begin{gathered}
\hat V'_{[k]} =   C_m T_{n-2} \nabla_{\boldsymbol{\sigma}_{n-3}} T_{n-3} \nabla_{\boldsymbol{\sigma}_{n-4}} T_{n-4}\cdots  
\nabla_{\boldsymbol{\sigma}_{k+1}}  T_{k+1} ,
\\
T_l= B'\lb\bs\sigma_{l} \rb F^{-1}\lb \bs\sigma_{l} ,a_l,\boldsymbol{\sigma}_{l-1}\rb B^{-1}\lb\bs\sigma_{l-1} \rb ,
\end{gathered}
\]
and $\bs\sigma_{n-2}:=-\bs{\theta}_{n-1} +\bs h_m$. Although it may seem that $\hat V'_{[k]}$ depends on $m$, thanks to the identity
\be\label{FBCFB}
\begin{gathered}
 C_m B'\lb -\boldsymbol{\theta}_{n-1} +\bs h_m \rb F^{-1}\lb -\boldsymbol{\theta}_{n-1} +\bs h_m,a_{n-2},\boldsymbol{\sigma}_{n-3}\rb =
 \\
= -B\lb  -\boldsymbol{\theta}_{n-1} \rb F^{-1}\lb -\boldsymbol{\theta}_{n-1},a_{n-2}-1,\boldsymbol{\sigma}_{n-3}\rb,
\end{gathered}
\ee
this dependence actually disappears. One may also use the identity $B^{-1}\lb\bs\sigma_{l} \rb \nabla_{\bs{\sigma}_l}  B'\lb \bs\sigma_{l} \rb = \nabla_{\bs{\sigma}_l}$ (cf (\ref{BBp}))
to simplify the expression (\ref{MVBV}) for the operator-valued monodromy matrix $\hat M_{[k]}$:
\be\label{MVBVs}
\hat M_{[k]} = \hat V_{[k]} B^{2}\lb\boldsymbol{\sigma}_{k}\rb \hat V_{[k]}^{-1},
\ee
\[
\hat V_{[k]} =  B\lb -\boldsymbol{\theta}_{n-1} \rb F^{-1}\lb -\boldsymbol{\theta}_{n-1},a_{n-2}-1,\boldsymbol{\sigma}_{n-3}\rb  
\nabla_{\boldsymbol{\sigma}_{n-3}} F^{-1}\lb \boldsymbol{\sigma}_{n-3},a_{n-3},\boldsymbol{\sigma}_{n-4}\rb \times
\]\[
\times 
\nabla_{\boldsymbol{\sigma}_{n-4}} F^{-1}\lb \boldsymbol{\sigma}_{n-4},a_{n-4},\boldsymbol{\sigma}_{n-5}\rb \nabla_{\boldsymbol{\sigma}_{n-5}}
\cdots  
\nabla_{\boldsymbol{\sigma}_{k+1}} F^{-1}\lb \boldsymbol{\sigma}_{k+1},a_{k+1},\boldsymbol{\sigma}_{k}\rb .
\]
Recall that there is a simple formula (\ref{FiF}) for $F^{-1}$.

The operator-valued monodromy matrix $\hat M_{[k]}$ has periodic dependence on $\bs\sigma=\lb\bs\sigma_1,\ldots,\bs\sigma_{n-3}\rb$.
Namely, $\hat M_{[k]}$ is invariant under the shifts of
any $\bs \sigma_l$, $l=1,\ldots,n-3$, by the vectors $\bs w_l$ of root lattice $\mathfrak R$ of  $\mathfrak{sl}_N$ embedded into $\mathbb{C}^N$, i.e.
$\bs w_l \in \mathbb{Z}^N$, $\sum_{s=1}^N w_l^{(s)}=0$.  
This claim is equivalent to the following proposition.

\begin{prop}\label{Periodicity}
The operator-valued monodromy matrix $\hat M_{[k]}$ has the following property:
\[
 \nabla_{\bs\sigma_l,m}  \nabla_{\bs\sigma_l,s}^{-1}  \hat M_{[k]} = 
 \hat M_{[k]} \nabla_{\bs\sigma_l,m}  \nabla_{\bs\sigma_l,s}^{-1} \,.
\]
\end{prop}

\begin{proof} For $l<k$, $\hat M_{[k]}$ does not depend on $\bs \sigma_l$. Therefore it is sufficient to prove 
the commutativity for the cases with $l\ge k$. 

For $l=k$, let us extract explicitly the part of $\hat M_{[k]}$  given by (\ref{MVBVs})
depending on $\bs\sigma_k$:
\[
\hat M_{[k]} = \cdots  F^{-1}\lb \boldsymbol{\sigma}_{k+1},a_{k+1},\boldsymbol{\sigma}_{k}\rb  B^{2}\lb \boldsymbol{\sigma}_{k}\rb
F\lb \boldsymbol{\sigma}_{k+1},a_{k+1},\boldsymbol{\sigma}_{k}\rb \cdots.
\]
Using the identities (\ref{shiftroot}) and   $B^{2}\lb\boldsymbol{\sigma}_{k}\rb=\exp \lb 2\pi i \operatorname{diag}\bs\sigma_k\rb$,  we then obtain 
\[\begin{gathered}
\nabla_{\bs\sigma_k,m}  \nabla_{\bs\sigma_k,s}^{-1}  \hat M_{[k]} =
\cdots  F^{-1}\lb \boldsymbol{\sigma}_{k+1},a_{k+1},\boldsymbol{\sigma}_{k}+\bs h_m-\bs h_s\rb  B^{2}\lb\boldsymbol{\sigma}_{k}+\bs h_m-\bs h_s\rb \times\\
\times F\lb \boldsymbol{\sigma}_{k+1},a_{k+1},\boldsymbol{\sigma}_{k}+\bs h_m-\bs h_s\rb  \cdots
 \nabla_{\bs\sigma_k,m}  \nabla_{\bs\sigma_k,s}^{-1} = \hat M_{[k]} 
 \nabla_{\bs\sigma_k,m}  \nabla_{\bs\sigma_k,s}^{-1}  .
 \end{gathered}
\]
For $l>k$, it suffices to show the commutativity property for  $\hat V_{[k]}$: 
\be\label{compropV}
\nabla_{\bs\sigma_l,m}  \nabla_{\bs\sigma_l,s}^{-1}  \hat V_{[k]} = 
\hat V_{[k]} \nabla_{\bs\sigma_l,m}  \nabla_{\bs\sigma_l,s}^{-1}\,.
\ee
Indeed, the part of the operator-valued matrix  $\hat V_{[k]}$  
which depends on $\bs\sigma_l$ is given by
\[
\hat V_{[k]} =\cdots 
F^{-1}\lb \boldsymbol{\sigma}_{l+1},a_{l+1},\boldsymbol{\sigma}_l \rb
\nabla_{\boldsymbol{\sigma}_l} F^{-1}\lb \boldsymbol{\sigma}_l,a_l,\boldsymbol{\sigma}_{l-1}\rb\cdots .
\]
The  property (\ref{compropV}) easily follows from (\ref{shiftroot}) and the diagonal form of of $\nabla_{\boldsymbol{\sigma}_l}$.
\end{proof}

Although $\mathcal{F}_m\lb \bs\sigma|\,y\rb$ is not invariant with respect to the shifts of $\bs \sigma_l$ by root vectors, 
the above proposition suggests to consider Fourier transform of $\mathcal{F}_m\lb \bs\sigma|\,y\rb$ with respect to such shifts.
It will be shown in the next section that this Fourier transform gives the solution of the Riemann-Hilbert problem we are interested in.

\section{Solution of semi-degenerate Fuchsian system\label{sec5} \\ and isomonodromic  tau function}

As above, let $\mathfrak R$ denote the root lattice of $\mathfrak{sl}_N$ embedded into $\mathbb{C}^N$, i.e. $\bs w \in \mathfrak R$ if and only if 
$\bs w \in \oplus_{k=1}^{N-1} \mathbb{Z} \lb \bs h_k-\bs h_{k+1}\rb$, where $\bs h_k$ are the weights  of the first fundamental representation of $\mathfrak{sl}_N$, 
cf (\ref{hweight}).
Equivalently, $\bs w\in \mathfrak R$ if and only if $\bs w \in \mathbb{Z}^N$ and the sum of its components vanishes.

\begin{theo} 	\label{solinfty}
Let $\Phi\lb y\rb$ be the solution of the semi-degenerate Fuchsian system (\ref{FS}) having the asymptotics $\Phi\lb y\rb=y^{-\Theta_{n-1}}\bigl[\mathbb I+O\lb y^{-1}\rb\bigr]$ as $y\to\infty$ and
monodromies $\left\{M_k\right\}_{k=0,\ldots,n-1}$, described by Proposition~\ref{propparam} with 
parameters 
$\bt_0$, $\bt_{n-1}$; $\left\{\bs r_k,a_k\right\}_{k=1,\ldots,n-2}$, $\left\{\bs\sigma_l\right\}_{l=1,\ldots,n-3}$.
The matrix elements of $\Phi(y)$ can be written in terms of conformal blocks of the $W_N$-algebra:
\be\label{solLPCBnorm}
\Phi_{jm}(y)= \frac{\lb-1\rb^{N(\delta_{jm}-1)}}{\mathcal {N}_m}  \frac{
 \left\langle -\boldsymbol{\theta}_{n-1} +\bs h_m \left| \psi_j\lb y\rb\right|\Theta^\mathrm{D}_{jm}\right\rangle} 
 {\left\langle -\boldsymbol{\theta}_{n-1}|\,\Theta^\mathrm{D}\right\rangle},
\ee
where
\begin{align}\nonumber
& \left|\Theta_{jm}\lb\boldsymbol{\sigma}_1,\ldots,\boldsymbol{\sigma}_{n-3}\rb\right\rangle:=\\
\nonumber
 &\qquad\;\; =\mathcal{P}_{ -\boldsymbol{\theta}_{n-1} -\bs h_j+\bs h_m} V_{a_{n-2}}\lb z_{n-2}\rb
 \mathcal{P}_{\boldsymbol{\sigma}_{n-3}}
 V_{a_{n-3}}\lb z_{n-3}\rb\mathcal{P}_{\boldsymbol{\sigma}_{n-4}} \cdots \mathcal{P}_{\boldsymbol{\sigma}_1}
 V_{a_{1}}\lb z_{1}\rb\left |\boldsymbol{\theta}_0\right\rangle,\\
\label{CBD}
 &\left|\Theta^\mathrm{D}_{jm}\right\rangle:=\sum_{\boldsymbol{w}_1,\ldots,\boldsymbol{w}_{n-3}\in \mathfrak R} 
 e^{(\boldsymbol{\beta}_1,\boldsymbol{w}_1)+\cdots+(\boldsymbol{\beta}_{n-3},\boldsymbol{w}_{n-3})}
 \left|\Theta_{jm}\lb\boldsymbol{\sigma}_1+\boldsymbol{w}_{1},\ldots,\boldsymbol{\sigma}_{n-3}+\boldsymbol{w}_{n-3}\rb\right\rangle,\\
\nonumber
 &\left|\Theta^\mathrm{D}\right\rangle:=|\Theta^\mathrm{D}_{11}\rangle=\ldots =\left|\Theta^\mathrm{D}_{NN}\right\rangle,
\end{align}
and $\mathfrak R$ is the root lattice of $\mathfrak{sl}_N$. 
The relation between the monodromy parameters $\bs r_k$ and conjugate Fourier momenta $\bs \beta_k$ is given by the formulas (\ref{xME}), (\ref{yME}) of Proposition~\ref{XYME} and
\be\label{betakrk}
\begin{gathered}
R_k=Y_{k}^{-1}H_{k}X_{k+1}, \quad 1\le k \le n-3, \qquad 
R_{n-2}=Y_{n-2}^{-1} B\lb\boldsymbol{\theta}_{n-1} \rb ,
\\
R_k=\operatorname{diag}\bs r_k,\qquad H_k=\operatorname{diag}\bigl(e^{\beta_k^{(1)}},\ldots,e^{\beta_k^{(N)}}\bigr).\end{gathered}
\ee
The normalization coefficients $\mathcal{N}_m$ follow from  the normalization (\ref{normpsi}) of $\psi_m\lb y\rb$:  
\[
\mathcal{N}_m=\check{\cN}\lb-\bt_{n-1}+\bs h_m,\bs h_1,-\bt_{n-1}\rb. 
\]
\end{theo}

\begin{proof}
We have to show that proposed matrix $\Phi(y)$ solves the initial Riemann--Hilbert problem with the monodromy matrices specified above. 
The monodromy matrix $M_{n-1}$ around $z_{n-1}=\infty$ in the positive direction on the Riemann sphere is diagonal and has  matrix elements 
\[
\lb M_{n-1}\rb_{jj}=
\bigl( {B'}^{-2}\lb -\boldsymbol{\theta}_{n-1} +\bs h_m \rb\bigr)_{jj}=\exp\left\{2\pi i\lb {\theta}_{n-1}^{(j)} -h_m^{(j)}+\tfrac{N-1}{N} \rb \right\}=
e^{2\pi i {\theta}_{n-1}^{(j)}}.
\]

The next monodromy matrix is still numerical,
\[
\hat M_{[n-3]}=
B\lb -\boldsymbol{\theta}_{n-1} \rb F^{-1}\lb -\boldsymbol{\theta}_{n-1},a_{n-2}-1,\boldsymbol{\sigma}_{n-3}\rb
 B^{2}\lb \boldsymbol{\sigma}_{n-3}\rb
F\lb  -\boldsymbol{\theta}_{n-1},a_{n-2}-1,\boldsymbol{\sigma}_{n-3}\rb
B^{-1}\lb -\boldsymbol{\theta}_{n-1} \rb.
\]
The next one, however, has operator-valued entries:
\[
\begin{gathered}
\hat M_{[n-4]}=     
B\lb -\boldsymbol{\theta}_{n-1}\rb F^{-1}\lb -\boldsymbol{\theta}_{n-1},a_{n-2}-1,\boldsymbol{\sigma}_{n-3}\rb
\nabla_{\boldsymbol{\sigma}_{n-3}}
F^{-1}\lb\boldsymbol{\sigma}_{n-3},a_{n-3},\boldsymbol{\sigma}_{n-4}\rb\times\\
\times\,
 B^{2}\lb\boldsymbol{\sigma}_{n-4}\rb F\lb\boldsymbol{\sigma}_{n-3},a_{n-3},\boldsymbol{\sigma}_{n-4}\rb \nabla^{-1}_{\boldsymbol{\sigma}_{n-3}}
F\lb -\boldsymbol{\theta}_{n-1},a_{n-2}-1,\boldsymbol{\sigma}_{n-3}\rb
B^{-1}\lb -\boldsymbol{\theta}_{n-1}\rb.
\end{gathered}
\]
This operator expression is invariant with respect to the shifts of intermediate weights $\bs\sigma_{n-3}$ and 
$\bs\sigma_{n-4}$ by root vectors $\bs h_m-\bs h_s$ of $\mathfrak{sl}_N$ (see Proposition~\ref{Periodicity})
and therefore it is the same for all conformal blocks appearing in (\ref{CBD}).

Let us move the components of $\nabla_{\bs\sigma}$ to the right to act on conformal blocks.
We need a few identities which follow from (\ref{shiftid}) and (\ref{shiftroot}) written in components,
\begin{align*}
 \nabla_{\bs\sigma,l} F_{lj}\lb-\bs\sigma',a,-\bs\sigma\rb=&\,
 (-1)^{N-1} F_{lj}\lb-\bs\sigma',a-1,-\bs\sigma\rb \nabla_{\bs\sigma,l},
\\
 \nabla_{\bs\sigma,m}  F_{jl}\lb\bs\sigma,a,\bs\sigma'\rb\;\;=&\,
 (-1)^{(N-1)\delta_{ml}} F_{jl}\lb\bs\sigma,a-1,\bs\sigma'\rb \nabla_{\bs\sigma,m} ,
\\
 \nabla_{\bs\sigma,m} \nabla_{\bs\sigma,l}^{-1} F_{lp}\lb\bs\sigma'',a,\bs\sigma\rb=&\,
 (-1)^{(N-1)(\delta_{ml}-1)} F_{lp}\lb\bs\sigma'',a,\bs\sigma\rb \nabla_{\bs\sigma,m} \nabla_{\bs\sigma,l}^{-1},
\end{align*}
to obtain
\[\begin{gathered}
\hat M_{[n-4]}=      
B\lb -\boldsymbol{\theta}_{n-1} \rb F^{-1}\lb -\boldsymbol{\theta}_{n-1},a_{n-2}-1,\boldsymbol{\sigma}_{n-3}\rb
\tilde \nabla_{\boldsymbol{\sigma}_{n-3}}
F^{-1}\lb\boldsymbol{\sigma}_{n-3},a_{n-3}-1,\boldsymbol{\sigma}_{n-4}\rb\times
\\
\times
 B^{2}\lb \boldsymbol{\sigma}_{n-4}\rb F\lb \boldsymbol{\sigma}_{n-3},a_{n-3}-1,\boldsymbol{\sigma}_{n-4}\rb \tilde \nabla^{-1}_{\boldsymbol{\sigma}_{n-3}}
F\lb -\boldsymbol{\theta}_{n-1},a_{n-2}-1,\boldsymbol{\sigma}_{n-3})B^{-1}( -\boldsymbol{\theta}_{n-1} \rb,
\end{gathered}
\]
where the components of $\tilde \nabla_{\boldsymbol{\sigma}_{n-3}}$ act directly on conformal blocks.
Therefore, the action on the Fourier transformed conformal blocks diagonalizes producing a numerical monodromy matrix, 
\[
\begin{gathered}
M_{[n-4]}=B\lb -\boldsymbol{\theta}_{n-1} \rb F^{-1}\lb -\boldsymbol{\theta}_{n-1},a_{n-2}-1,\boldsymbol{\sigma}_{n-3}\rb
H_{n-3}^{-1} 
F^{-1}\lb \boldsymbol{\sigma}_{n-3},a_{n-3}-1,\boldsymbol{\sigma}_{n-4}\rb
\times \\
\times\, 
 B^{2}\lb\boldsymbol{\sigma}_{n-4}\rb  F\lb\boldsymbol{\sigma}_{n-3},a_{n-3}-1,
 \boldsymbol{\sigma}_{n-4}\rb  H_{n-3}
F\lb -\boldsymbol{\theta}_{n-1},a_{n-2}-1,
\boldsymbol{\sigma}_{n-3}\rb B^{-1}\lb -\boldsymbol{\theta}_{n-1} \rb,
\end{gathered}
\]
where $H_{k}$ is the diagonal matrix with the elements $\lb H_{k}\rb_{mj}=e^{\beta_k^{(j)}} \delta_{mj}$. 

Analogous procedure may be applied to other monodromy matrices. We note that their structure as products of elementary building blocks reproduces the structure of the corresponding
products (\ref{MkWk}), (\ref{WWW}) in Proposition~\ref{propparam}. The exact relation \eqref{betakrk} between the parameters is obtained using (\ref{FW}).
\end{proof}

There is a related theorem for an alternative normalization of the fundamental solution.

\begin{theo}\label{soly0}
The solution of the semi-degenerate Fuchsian system with the same monodromies, normalized as $\Phi\lb y_0\rb = \mathbb{I}$, can be written in terms of conformal blocks as
\be\label{solLPCB}
\Phi_{jm}\lb y\rb=\lb-1\rb^{N\lb\delta_{jm}-1\rb} \lb y_0-y\rb^{(N-1)/N}  \frac{
 \left\langle -\boldsymbol{\theta}_{n-1}\left| \bar \psi_m( y_0)\, \psi_j\lb y\rb\right|\Theta^\mathrm{D}_{jm}\right\rangle }
 {\left\langle -\boldsymbol{\theta}_{n-1}|\, \Theta^\mathrm{D}\right\rangle},
 \ee
where $\left|\Theta^\mathrm{D}_{jm}\right\rangle$ and $\left|\Theta^\mathrm{D}\right\rangle$ are given by the same formulas (\ref{CBD}).
The normalizations of $ \psi_j(y)$ and $\bar \psi_m(y_0)$ are fixed by (\ref{normpsi}) and (\ref{normpsibar}). 
\end{theo}

The proof of Theorem~\ref{soly0} goes along the same lines as the proof of Theorem~\ref{solinfty}. 
Note that the solution \eqref{solLPCBnorm} can be obtained from the solution (\ref{solLPCB}) by fusion of  $\langle -\boldsymbol{\theta}_{n-1}| \bar \psi_m(y_0)$ in the limit $y_0\to \infty$.

By the arguments given in \cite{GIL1}, one obtains the isomonodromic tau function (\ref{isomtau}) as the  Fourier transform of the $W_N$-conformal blocks which appears in the denominators of the fundamental solutions (\ref{solLPCBnorm}) and (\ref{solLPCB}):
\begin{prop}\label{proptaucb} The isomonodromic tau function of Jimbo-Miwa-Ueno for semi-degenerate Fuchsian systems is given by 
\be\label{tau_CB}
\tau(\bs z) = \left\langle -\boldsymbol{\theta}_{n-1}|\, \Theta^\mathrm{D}\right\rangle,
\ee
where $| \Theta^\mathrm{D}\rangle$ is given by (\ref{CBD}). 
\end{prop}

Let us exemplify the above results with the simplest example of $n=4$, i.e. two generic punctures and two punctures of spectral type $\lb N-1,1\rb$. The former are located at $z_0=0$ and $z_3=\infty$. It can be also assumed that $z_2=1$, so that the only remaining time parameter is $z_1\equiv t$. The Fuchsian system \eqref{FS} then reduces to 
\be\label{FSFST}
\frac{d\Phi\lb y\rb}{dy} = \Phi\lb y\rb \lb\frac{A_0}{y}+\frac{A_t}{y-t}+\frac{A_1}{y-1}\rb.
\ee
It should be stressed that the matrices $A_{0}$, $A_t$, $A_1$ are \textit{traceless} (this assumption involves no loss of generality but is crucial from the CFT perspective). Their respective diagonalizations are
$$\operatorname{diag}\bs\theta_{0},\qquad \operatorname{diag}\lb\tfrac{N-1}{N},-\tfrac 1N,\ldots,-\tfrac 1N\rb a_t,\qquad
\operatorname{diag}\lb\tfrac{N-1}{N},-\tfrac 1N,\ldots,-\tfrac 1N\rb a_1,$$ 
with $a_t,a_1\in\Cb$. 
We also have $A_{\infty}=-A_0-A_t-A_1=\operatorname{diag}\bs\theta_{\infty}$. The monodromy preserving deformations for \eqref{FSFST} are described by the equations
\be
\frac{dA_0}{dt}=\frac{[A_0,A_t]}{t},\qquad \frac{dA_1}{dt}=\frac{[A_1,A_t]}{t-1},
\ee
equivalent to the polynomial Hamiltonian Fuji-Suzuki-Tsuda system  \cite{Tsuda}. The corresponding tau function is defined, up to arbitrary non-zero constant factor, by
\be
\frac{d\ln\tau_{_\mathrm{FST}}\lb t\rb}{dt}=\frac{\operatorname{Tr}A_0A_t}{t}+\frac{\operatorname{Tr}A_1A_t}{t-1}.
\ee

Proposition \ref{proptaucb} states that $\tau_{_\mathrm{FST}}\lb t\rb$ is nothing but the Fourier transform of the 4-point semi-degenerate conformal block,
\be\label{FST4p}
\tau_{_\mathrm{FST}}\lb t\rb=
\operatorname{const}\cdot\sum_{\boldsymbol{w}\in \mathfrak R} 
  e^{(\boldsymbol{\beta},\boldsymbol{w})} \begin{tikzpicture}[baseline,yshift=-0.4cm,scale=1.2]
   \draw [thick] (0,0) -- (2.6,0);
   \draw [thick,dashed] (0.8,0) -- (0.8,0.8); 
   \draw [thick,dashed] (1.8,0) -- (1.8,0.8);
   \draw (1.35,0) node[below] {\scriptsize $\bs{\sigma}+\bs w$};
   \draw (0.5,0) node[below] {\scriptsize $-\bs{\theta}_{\infty}$}; 
   \draw (2.3,0) node[below] {\scriptsize $\bs{\theta}_{0}$}; 
   \draw (0.8,0.8) node[above] {$1$};
   \draw (0.8,0.4) node[left] {\scriptsize ${a_{1}}$};
   \draw (1.8,0.4) node[right] {\scriptsize ${a_{t}}$};
   \draw (1.8,0.8) node[above] {$t$};
   \draw (0,0) node[left] {$\infty$};    
   \draw (2.6,0) node[right] {$0$};  
   \end{tikzpicture}
\ee
 where the  vertices represent appropriate chiral vertex operators in 
 \eqref{CBD}. The parameters $\bs\sigma=\bigl(\sigma^{(1)},\ldots,\sigma^{(N)}\bigr)\in\Cb^N$ with $\sum_{s=1}^N\sigma^{(s)}=0$ are related to diagonalized composite monodromy of $\Phi \lb y\rb$ around $0$ and $t$: $M_0M_t\sim e^{2\pi i\operatorname{diag}\bs \sigma}$. 
 
 There is a freedom in the choice of $\bs\sigma$, which can be shifted by any vector in $\mathfrak R$. Obviously, such a shift does not affect  the result \eqref{FST4p}. From the point of view of asymptotic analysis of $\tau_{_\mathrm{FST}}\lb t\rb$ as $t\to 0$, it is convenient to choose $\bs\sigma$ so that $\left|\Re \sigma^{(i)}-\Re\sigma^{(j)}\right|\le 1$ for $i,j=1,\ldots,N$. Indeed,
 \be\label{FSTas1}
 \tau_{_\mathrm{FST}}\lb t\rb\simeq t^{\frac12\lb \bs\sigma^2-\bs\theta_0^2-a_t^2\bs h_1^2\rb}
 \sum_{\bs w\in\mathfrak R}\mathcal C\lb\bs\sigma,\bs w\rb e^{(\boldsymbol{\beta},\boldsymbol{w})}t^{\frac{\bs w^2}{2}+\lb\bs\sigma,\bs w\rb}\Bigl[1+O\lb t\rb\Bigr],
 \ee
  as $t\to 0$, where
 \be\mathcal C\lb\bs\sigma,\bs w\rb=\frac{\mathcal N\lb -\bs\theta_{\infty},a_1\bs h_1,\bs\sigma+\bs w\rb\mathcal N\lb \bs\sigma+\bs w,a_t\bs h_1,\bs\theta_0\rb}{\mathcal N\lb -\bs\theta_{\infty},a_1\bs h_1,\bs\sigma\rb\mathcal N\lb \bs\sigma,a_t\bs h_1,\bs\theta_0\rb}.
 \ee
 The normalization coefficients $\mathcal N\lb \bs \sigma',a\bs h_1,\bs\sigma\rb$ are given by \eqref{Nnormcoef}, which implies that the structure constants $\mathcal C\lb\bs\sigma,\bs w\rb$ can be expressed in terms of gamma functions.
 In the generic case of strict inequality $\left|\Re \sigma^{(i)}-\Re\sigma^{(j)}\right|< 1$, the leading asymptotic contribution to \eqref{FSTas1} is determined by  $\bs w=0$, and the subleading terms correspond to the roots, $\bs w=\bs h_i-\bs h_j$. It follows that
 \begin{align}
 &\label{FSTas2}\tau_{_\mathrm{FST}}\lb t\rb\simeq t^{\frac12\lb \bs\sigma^2-\bs\theta_0^2-a_t^2\bs h_1^2\rb}\left[1+\sum_{i\ne j}^N\mathcal C\lb\bs\sigma,\bs h_i-\bs h_j\rb e^{\beta^{(i)}-\beta^{(j)}}t^{1+\sigma^{(i)}-\sigma^{(j)}}+
 \mathcal E\lb t\rb \right],\\
 &\mathcal E\lb t\rb=O\lb t\rb +\sum_{i,j,k,l=1}^N O\lb t^{2+\Re\lb \sigma^{(i)}-\sigma^{(j)}+\sigma^{(k)}-\sigma^{(l)}\rb}\rb.
 \end{align}
 The asymptotics \eqref{FSTas2} is  a specialization to the semi-degenerate case
 of the Proposition 3.9 of \cite{ILP}, which provides a higher-rank analogue of the Jimbo's asymptotic formula \cite{Jimbo} for Painlev\'e~VI.

\section{Discussion}
We conclude with some open research directions. As already mentioned in the introduction, the fundamental solutions of $3$-point Fuchsian systems allow to construct  Fredholm determinant representation of the $n$-point tau function \cite{CGL,GL16}.
 The principal minor expansion of the determinant gives a series representation for $\tau\lb \bs z\rb$, which in the semi-degenerate case may be expected to coincide with Nekrasov formulas \cite{Nekr} for instanton partition functions of linear quiver gauge theories in the self-dual $\Omega$-background. This would produce a new (direct) proof of the AGT-W relation \cite{AGT,Wyl,FL3,MM} for $c=N-1$. 
 
 Let us note that the combinatorics of the tau function expansions is the same in the generic and semi-degenerate case. Any progress on the 3-point solutions would provide new information on more general vertex operators and conformal blocks for $W_N$-algebras.
 
 It might be possible to adapt the technique developed in \cite{ILP} to compute the connection coefficient for the tau function of the Fuji-Suzuki-Tsuda system (relative normalization of the asymptotics as $t\to0$ and as $t\to\infty$). The CFT counterpart of this quantity is the fusion kernel  relating the $s$- and $u$-channel semi-degenerate $c=N-1$ conformal blocks (see \cite{ILT13} for $N=2$ case):
 $$
 \begin{tikzpicture}[baseline,yshift=-0.3cm,scale=1.3]
  \draw [thick] (0,0) -- (2.6,0);
\draw [thick,dashed] (0.8,0) .. controls +(0,0.5) and +(0,-0.5) .. (1.8,0.8);
    \draw [very thick,color=white] (1.25,0.38) -- (1.35,0.42); 
  \draw [thick,dashed] (1.8,0) .. controls +(0,0.5) and +(0,-0.5) .. (0.8,0.8); 
  \draw (1.4,0) node[below] {\scriptsize $\bs{\sigma}$};
  \draw (0.5,0) node[below] {\scriptsize $\bs{\theta}_{\infty}$}; 
  \draw (2.3,0) node[below] {\scriptsize $\bs{\theta}_{0}$}; 
  \draw (0.8,0.8) node[above] {$z_{2}$};
  \draw (1.8,0.8) node[above] {$z_{1}$};
  \draw (0.85,0.6) node[left] {\scriptsize ${a_{2}}$};
  \draw (1.75,0.6) node[right] {\scriptsize ${a_{1}}$};
  \draw (0,0) node[left] {$\infty$};    
  \draw (2.6,0) node[right] {$0$};  
  \end{tikzpicture}=
 \int d\bs{\sigma'}\; F\Bigl[\text{\small $\begin{array}{ll}a_2 & a_1 \\ \bs\theta_{\infty} &\bs\theta_0\end{array};\begin{array}{l}
 \bs\sigma' \\ \bs\sigma
 \end{array}$}\Bigr]\;
 \begin{tikzpicture}[baseline,yshift=-0.3cm,scale=1.3]
 \draw [thick] (0,0) -- (2.6,0);
 \draw [thick,dashed] (0.8,0) -- (0.8,0.8); 
 \draw [thick,dashed] (1.8,0) -- (1.8,0.8);
 \draw (1.4,0) node[below] {\scriptsize $\bs{\sigma'}$};
 \draw (0.5,0) node[below] {\scriptsize $\bs{\theta}_{\infty}$}; 
 \draw (2.3,0) node[below] {\scriptsize $\bs{\theta}_{0}$}; 
 \draw (0.8,0.8) node[above] {$z_{2}$};
 \draw (0.8,0.4) node[left] {\scriptsize ${a_{2}}$};
 \draw (1.8,0.4) node[right] {\scriptsize ${a_{1}}$};
 \draw (1.8,0.8) node[above] {$z_{1}$};
 \draw (0,0) node[left] {$\infty$};    
 \draw (2.6,0) node[right] {$0$};  
 \end{tikzpicture}$$
 Besides the explicit evaluation of $F\left[\ldots\right]$, it would be interesting to clarify the relation of this quantity to symplectic geometry of the moduli space of semi-degenerate monodromy data.
 
 Monodromy preserving deformations and Fuchsian systems are also related to quasi-classical ($c\to\infty$) limit of CFT \cite{Teschner2010,NagoyaYamada,LLNZ, Teschner2017}. One may wonder whether a direct connection between the quasi-classical and $c=N-1$ $W_N$-conformal blocks can be established. Quantitative aspects of this relation for $N=2$ have been the subject of recent study \cite{LencsesNovaes}.

\appendix
\section{Singular vectors of semi-degenerate Verma modules}

For $c=N-1$, we are going to use a free-boson realization of the $W_N$-algebra to find the singular vectors of the semi-degenerate Verma modules.
It will be helpful to extend $W_N=W(\mathfrak{sl}_N)$ to $W(\mathfrak{gl}_N)$ by introducing one more free-boson field $J\lb z\rb$ with the OPE
\[
J\lb z\rb J\lb w\rb=\frac{1/N}{\lb z-w\rb^2}+\text{regular},
\]
and regular OPEs with the currents $J_k\lb z\rb$ entering the definition of $W_N$: $J_k\lb z\rb J\lb w\rb=\text{regular}$.
Introduce the currents $\tilde{J}_k(z)=J_k(z)+J(z)$ which have the following OPEs:
\[
\tilde{J}_k\lb z\rb \tilde{J}_l\lb w\rb=\frac{\delta_{kl}}{\lb z-w\rb^2}+\text{regular},
\]
and define the generators of the $W(\mathfrak{gl}_N)$-algebra:
\be\label{tWgen}
\widetilde W^{(s)}\lb z\rb=\sum_{1\le i_1<\ldots<i_s\le N} :\tilde J_{i_1}\lb z\rb \cdots \tilde J_{i_s}\lb z\rb:\,, \qquad s=1,\ldots,N.
\ee
We will use the mode expansion
\be\label{tJNmodes}
\tilde J_k\lb z\rb=\sum_{p\in \mathbb{Z}} \frac{\tilde a^{(k)}_p}{z^{p+1}},
\ee
with modes acting on the Fock space $\mathcal{F}$ generated from the vacuum state $\left|\bt\right\ra$, $\bt=\lb \theta_1,\ldots,\theta_N\rb\in \mathbb{C}^N$:
\[
\tilde a^{(k)}_0 \left|\bt\right\ra=\theta_k \left|\bt\right\ra, \qquad \tilde a^{(k)}_p \left|\bt\right\ra=0,\qquad k=1,\ldots,N,\quad p>0.
\]

For the semi-degenerate representation of the $W$-algebra, we use $\bt=\lb a,0,\ldots,0\rb$. In what follows,  only such $\bt$ will be used.
For the action of modes of $\widetilde W^{(s)}\lb z\rb$ defined by 
\[
\widetilde{W}^{(s)}\lb z\rb=\sum_{p\in \mathbb{Z}} \frac{\widetilde{W}^{(s)}_p}{z^{p+s}},
\]
we have 
\be\label{Wth0}
\widetilde{W}^{(s)}_{-p} \left|\bt\right\ra= 0,\qquad 0\le p \le s-1\,.
\ee
{}On the other hand, we have a relation which can be obtained from the expansion
of (\ref{tWgen}) with the use of (\ref{Wgen}):
\be\label{reltWW}
\widetilde{W}^{(s)}\lb z\rb=\sum_{r=0}^s \binom{N-r}{N-s} :J^{s-r}\lb z\rb: W^{(r)}\lb z\rb.
\ee

In order to find the relations for the elements of  $W_N=W(\mathfrak{sl}_N)$ in the Fock module $\mathcal{F}$, let us use the relation
(\ref{reltWW}) modulo vectors from the submodule $\mathcal{F}'$ generated by $a_{-p}|\bt\ra$, $p>0$. 
We have, for $p>0$,
\be\label{reltWWmodes}
\widetilde{W}^{(s)}_{-p}\left|\bt\right\ra = \sum_{r=0}^s \binom{N-r}{N-s} \lb a/N\rb^{s-r}  W^{(r)}_{-p}\left|\bt\right\ra \quad \text{mod}\  \mathcal{F}'.
\ee
These relations are linear and can be inverted:
\be\label{relWtWmodes}
{W}^{(r)}_{-p}\left|\bt\right\ra = \sum_{s=0}^r \binom{N-s}{N-r} \lb -a/N\rb^{r-s} \widetilde W^{(s)}_{-p}\left|\bt\right\ra \quad \text{mod}\  \mathcal{F}'.
\ee
Using (\ref{Wth0}), rewrite (\ref{relWtWmodes}) as
\be\label{relWtWmres}
\begin{aligned}
&{W}^{(r)}_{-p}\left|\bt\right\ra = \sum_{s=2}^{p+1} \binom{N-s}{N-r} \lb -a/N\rb^{r-s} \widetilde W^{(s)}_{-p}\left|\bt\right\ra \quad  \lb\text{mod}\  \mathcal{F}' \rb=\\
&=\sum_{s=2}^{p+1} \binom{N-s}{N-r} (-a/N)^{r-s} \sum_{t=2}^s 
 \binom{N-t}{N-s} \lb a/N\rb^{s-t}  
 W^{(t)}_{-p}\left|\bt\right\ra .
 \end{aligned}
\ee
After summation over $s$ and changing summation index $t$ to $s$, we finally get 
\be\label{semdegrel_app}
 \left(W^{(r)}_{-p}+(-1)^{r+p} \sum_{s=2}^{p+1}  {N-s\choose r-s} {r-s-1\choose p-s+1} \left(\frac{a}{N}\right)^{r-s} W^{(s)}_{-p}\right)|\bt\ra=0, 
 \qquad 2\le p+1<r\le N.
\ee

\section{Null vectors and fusion rules for completely\\ degenerate fields}

This appendix uses a free-fermionic realization of the extension of $W_N=W(\mathfrak{sl}_N)$ to $W(\mathfrak{gl}_N)$.
It is convenient since the fermionic fields realize completely degenerate fields for $W(\mathfrak{gl}_N)$ and their properties 
can be studied easily.

The algebra of $N$-component free-fermionic fields is generated by $\psi^+_\alpha\lb z\rb$, $\psi^-_\alpha\lb w\rb$, $\alpha=1, \ldots, N$, with the standard 
singular part of the OPEs:
\be
\psi^+_\alpha\lb z\rb \psi^-_\beta\lb w\rb\sim  \frac{\delta_{\alpha,\beta}}{z-w}, \qquad
\psi^+_\alpha\lb z\rb \psi^+_\beta\lb w\rb\sim 0,\qquad \psi^-_\alpha\lb z\rb \psi^-_\beta\lb w\rb\sim 0. 
\ee
The $W(\mathfrak{gl}_N)$-algebra is a subalgebra of the free-fermionic algebra generated by the fields
\be\label{hwk}
\hwk\lb z\rb=\sum_{\alpha=1}^N :\partial^{k-1}\psi^+_\alpha\lb z\rb\psi^-_\alpha\lb z\rb:\,, \qquad k=1,\ldots,N.
\ee
It is convenient to extend this definition to all integer $k>0$. Using the bosonization formulas $\psi^\pm_\alpha\lb z\rb={:\exp ({\pm i\phi_\alpha\lb z\rb}):}$, $\alpha=1,\ldots,N$,
we get a free-boson realization of  $\hwk\lb z\rb$  as differential polynomials in bosonic currents
$\tilde J_\alpha\lb z\rb=i \partial \phi_\alpha\lb z\rb$:
\be
\hwk\lb z\rb=\sum_{\alpha=1}^N:\partial^{k} e^{i\phi_\alpha\lb z\rb}\cdot e^{-i\phi_\alpha\lb z\rb}:\,.
\ee
This free-boson realization of the currents $\hwk(z)$ does not coincide with the currents $\widetilde W^{(k)}\lb z\rb$ given by the formula
(\ref{tWgen}) of the Appendix~A, but
they generate the same algebra (a proof of this fact can be found in e.g. \cite{BGM}). 

The OPE
\be\label{OPE_Wpsi}
\hwk\lb z\rb\psi^+_\alpha\lb w\rb\sim \frac{\partial^{k-1}\psi^+_\alpha\lb w\rb}{z-w} 
\ee
means that 
\be
\hwk_{1-k}\left|\psi^+_\alpha \right\rangle= L_{-1}^{k-1}\left|\psi^+_\alpha \right\rangle,\qquad
\hwk_{m}|\psi^+_\alpha\rangle=0,\quad m>1-k,
\ee
which gives us the list of convenient null-vectors for the degenerate fields $\psi^+_\alpha\lb z\rb$ (they can also be rewritten in terms of standard generators). Note that the OPEs (\ref{OPE_Wpsi}) are the same for all $\alpha$ and, in fact, give the OPEs of the $W_N$-algebra currents with 
the completely degenerate fields.

To find the fusion rules for $\psi^+_\alpha\lb z\rb$, consider the conformal block
\be
\Omega^{(k)}_{\alpha}\lb z\rb=\left\langle\boldsymbol\theta_\infty\right|\hwk\lb z\rb\psi^+_\alpha\lb 1\rb\left|{\bt}_0\right\rangle .
\ee
Since 
\be
\left\langle\bt_\infty\right|\psi^+_\alpha\lb t\rb\left|{\bt}_0\right\rangle=t^{\Delta_\infty-\Delta_0-\Delta_t}=t^{\frac12\lb \bt_\infty^2-\bt_0^2-1\rb},
\ee
we have
\be
\left\langle\bt_\infty\right|\partial^{k-1}\psi^+_\alpha\lb 1\rb\left|{\bt}_0\right\rangle=
\bigl[\tfrac{\bt_\infty^2-\bt_0^2-1}{2}\bigr]_{k-1}=:A_k,
\ee
where $\left[x\right]_k=x\lb x-1\rb\cdots\lb x-k+1\rb$ denotes the falling factorial. This fixes the singular part of $\Omega^{(k)}_{\alpha}\lb z\rb$ near $z=1$ because of (\ref{OPE_Wpsi}):
\be\label{Omegaz1}
\Omega^{(k)}_{\alpha}\lb z\rb=\frac{A_k}{z-1}+O\lb 1\rb\quad \text{as } z\to 1.
\ee
In order to compute the asymptotics of $\Omega^{(k)}_{\alpha}\lb z\rb$ near $z=0$ and $z=\infty$, one may use the mode expansion
\[
\hwk\lb z\rb=\sum_{n\in \mathbb{Z}} \frac{\hwk_n}{z^{n+k}},
\]
considered for $|z|<1$ and $|z|>1$, respectively.
It gives the asymptotics 
\be\label{Omegaz0}
\Omega^{(k)}_{\alpha}(z)=\frac{w_k}{z^k}+\frac{c_{k-1}}{z^{k-1}}+\cdots+\frac{c_1}{z}+O\lb 1\rb\quad \text{as } z\to 0,
\ee
\be\label{Omegazi}
\Omega^{(k)}_{\alpha}(z)=\frac{w'_k}{z^k}+O\lb z^{-k-1}\rb\quad \text{as } z\to \infty,
\ee
where $w_k$ and $w'_k$ are the eigenvalues $w_k=\frac 1 k \sum_\alpha{[\theta_{0}^{(\alpha)}]}_k$
and $w_k'=\frac 1 k \sum_\alpha{[\theta_{\infty}^{(\alpha)}]}_k$ of $\hwk_{0}$ acting on the two vacua\footnote{
Here is a computation of $w_k$. Introduce and use the generating function
\[
\begin{gathered}
\hat W\lb z,w\rb=\sum_{k=1}^{\infty} \frac{\lb z-w\rb^{k-1}}{\lb k-1\rb !} \hwk\lb w\rb=
\sum_{k,\alpha} \frac{\lb z-w\rb^{k-1}}{\lb k-1\rb!} :\partial^{k-1}\psi_\alpha^+\lb w\rb\psi^-_\alpha\lb w\rb:\;=\; :\sum_{\alpha}\psi^+_\alpha\lb z\rb\psi^-_\alpha\lb w\rb:\;=
\\
=\sum_\alpha\psi^+_\alpha\lb z\rb\psi^-_\alpha\lb w\rb-\frac N{z-w}=\sum_\alpha :e^{i\phi_\alpha\lb z\rb}:\,:e^{-i\phi_\alpha\lb w\rb}:-\frac N{z-w},
\\
\langle \hat W\lb z,w\rb :e^{-i\lb \bt,\boldsymbol\phi\lb\infty\rb\rb}:\,:e^{i\lb\bt,
\boldsymbol\phi\lb 0\rb\rb}:\rangle=\sum_\alpha\frac{z^{\theta_\alpha}w^{-\theta_\alpha}-1}{z-w}
=\sum_\alpha\sum_{k=1}^\infty\frac{\lb z-w\rb^{k-1}}{\lb k-1\rb !}w^{-k}\frac{\left[\theta^{(\alpha)}\right]_k}{k}.
\end{gathered}
\]
}. Now from the asymptotics \eqref{Omegaz1}--\eqref{Omegazi} of $\Omega^{(k)}_{\alpha}(z)$ near $0$, $1$ and $\infty$ follows an exact formula
\be
\Omega^{(k)}_{\alpha}\lb z\rb=\frac{w_k}{z^k}+\frac{c_{k-1}}{z^{k-1}}+\ldots+\frac{c_1}{z}+\frac{A_k}{z-1},
\ee
where
\be
c_1=\cdots=c_{k-1}=-A_k, \qquad w_k+A_k=w'_k.
\ee

The equations $w_k+A_k=w'_k$, rewritten explicitly as
\be\label{rel_FR}
\sum_{\alpha=1}^N{[\theta_{\infty}^{(\alpha)}]}_k-
k\bigl[\tfrac{\bt_\infty^2-\bt_0^2-1}{2}\bigr]_{k-1}-\sum_{\alpha=1}^N{[\theta_{0}^{(\alpha)}]}_k=0,
\ee
give restrictions on the possible values of $\bt_\infty$ in terms of $\bt_0$ (fusion rules). 
Indeed, using the identity $\left[x+1\right]_k-\left[x\right]_k=k\left[x\right]_{k-1}$, rewrite (\ref{rel_FR}) as
\be\label{symeqfusion}
\sum_{\alpha=0}^N\left[x_\alpha\right]_k=
\sum_{\alpha=0}^N\left[y_\alpha\right]_k,
\ee
where $k=1,2,\ldots$ and
\be
x_{\alpha}=\theta_{\infty}^{(\alpha)}, \qquad y_{\alpha}=\theta_{0}^{(\alpha)}, \qquad \alpha>0 , \\
x_0=\tfrac12\lb{\bs\theta_\infty^2-\bs\theta_0^2-1}\rb, \qquad y_0=\tfrac12\lb{\bs\theta_\infty^2-\bs\theta_0^2+1}\rb.
\ee
The equations (\ref{symeqfusion}) require the coincidence of all symmetric polynomials in $N+1$ variables on two 
sets: $X=\{x_0,x_1,\ldots,x_N\}$ and $Y=\{y_0,y_1,\ldots,y_N\}$. This is possible only if these two sets coincide. Since
$x_0\ne y_0$, it means that there exist $\alpha',\alpha''>0$ such that 
\be
\theta_{\infty}^{(\alpha'')}=\tfrac12\lb{\bs\theta_\infty^2-\bs\theta_0^2+1}\rb,\qquad 
\theta_{0}^{(\alpha')}=
\tfrac12\lb{\bs\theta_\infty^2-\bs\theta_0^2-1}\rb,
\ee
and the sets formed by all the other $\theta$'s coincide. We immediately deduce from these equations that 
\be\label{fusion_theta_pp}
\theta_{\infty}^{(\alpha'')}=\theta_{0}^{(\alpha')}+1\,.
\ee
Since the variables in the sets $X$ and $Y$ are not independent, one also has to check consistency of the obtained solution; indeed,
\be
\theta_{0}^{(\alpha')}=\tfrac12\bigl((\theta_{0}^{(\alpha')}+1)^2-(\theta_{0}^{(\alpha')})^2-1\bigr).
\ee

Finally, let us recall that the $W(\mathfrak{gl}_N)$-modules generated from $\left|\bt\right\rangle$ and $\left|\bt'\right\rangle$ are isomorphic 
if the components of $\bt'$ are obtained from those of $\bt$ by a permutation. From this point of view, 
the fusion rules (\ref{fusion_theta_pp}) can be rewritten as $N$ possible channels (labeled by $\alpha=1,\ldots,N$) of changing $\bt_0$ to obtain $\bt_\infty$:
\be\label{fusion_theta}
\theta_{\infty}^{(\alpha)}=\theta_{0}^{(\alpha)}+1\,, \qquad \theta_{\infty}^{(\beta)}=\theta_{0}^{(\beta)}\,, \quad \beta\ne \alpha\,.
\ee
Moreover, each of these fusion channels is realized by the fusion with $\psi^+_\alpha\lb z\rb$,  $\alpha=1,\ldots,N$. 
This claim becomes clear in the bosonized picture, where 
\be 
\psi^+_\alpha\lb z\rb=\; :e^{i\phi_\alpha\lb z\rb}:\, ,\qquad \left|\bt_0\right\ra=\; :e^{i\lb \bt_0,\bs\phi\lb0\rb\rb}:\left|\bs 0\right\ra.
\ee

Returning to $W_N=W(\mathfrak{sl}_N)$, we introduce the bosonic field $\phi\lb z\rb=N^{-1} \sum_{\alpha=1}^N \phi_\alpha\lb z\rb$ and correct the fermionic fields 
by $:\exp i\phi\lb z\rb:$ to obtain the fields
\be\label{slferm}
\psi_\alpha\lb z\rb =\; :e^{-i \phi\lb z\rb} \psi^+_\alpha\lb z\rb:\,,\qquad \bar \psi_\alpha\lb z\rb =\; :e^{i \phi\lb z\rb} \psi^-_\alpha\lb z\rb:\,.
\ee
They have the following fusion rules: 
\be\label{frpsi}
\left\la \bs\theta_\infty \right| \psi_\alpha\lb z\rb \left|\bs\theta_0\right\ra \ne 0 \quad \text{if and only if}\quad \bs\theta_\infty = \bs\theta_0 +\bs h_\alpha\,,
\ee\be\label{frpsibar}
\left\la \bs\theta_\infty \right| \bar \psi_\alpha\lb z\rb \left|\bs\theta_0\right\ra \ne 0 \quad \text{if and only if}\quad \bs\theta_\infty = \bs\theta_0 -\bs h_\alpha\,,
\ee
where $h_\alpha$, $\alpha=1,\ldots,N$, are the weights of the first fundamental representations of $\mathfrak{sl}_N$, with components
$h_\alpha^{(s)}=\delta_{\alpha,s}-1/N$.

\end{document}